\def\year{2021}\relax
\documentclass[letterpaper]{article} 
\usepackage{aaai21}  
\usepackage{times}  
\usepackage{helvet} 
\usepackage{courier}  
\usepackage[hyphens]{url}  
\usepackage{graphicx} 
\usepackage{natbib}  
\usepackage{caption} 
\usepackage{enumitem}
\usepackage{amsthm}
\usepackage[utf8]{inputenc} 
\usepackage{tikz,pgfplots}
\usepackage{import}
\usepackage{url}            
\usepackage{booktabs}       
\usepackage{amsfonts}       
\usepackage{nicefrac}       
\usepackage{microtype}      
\usepackage{amssymb,amsmath,graphicx,makecell,multirow,subcaption,algorithm,algorithmic}
\newcommand{\overbar}[1]{\mkern 1.5mu\overline{\mkern-1.5mu#1\mkern-1.5mu}\mkern 1.5mu}

\title{Tune-In: Training Under Negative Environments with Interference for Attention Networks Simulating Cocktail Party Effect}

\author {
        Jun Wang\textsuperscript{\rm 1}, 
        Max W. Y. Lam\textsuperscript{\rm 1}, 
        Dan Su\textsuperscript{\rm 1}, 
        Dong Yu \textsuperscript{\rm 2}\\
}
\affiliations {
    \textsuperscript{\rm 1} Tencent AI Lab, Shenzhen, China \\
    \textsuperscript{\rm 2} Tencent AI Lab, Bellevue WA, USA \\
    \{joinerwang, maxwylam, dansu, dyu\}@tencent.com
}
\newtheorem{definition}{Definition}[section]
\newtheorem{assumption}{Assumption}[section]
\newtheorem{claim}{Claim}[section]

\begin{document}
\maketitle

\begin{abstract}
We study the cocktail party problem and propose a novel attention network called Tune-In, abbreviated for training under negative environments with interference. It firstly learns two separate spaces of speaker-knowledge and speech-stimuli based on a shared feature space, where a new block structure is designed as the building block for all spaces, and then cooperatively solves different tasks. Between the two spaces, information is cast towards each other via a novel cross- and dual-attention mechanism, mimicking the bottom-up and top-down processes of a human's cocktail party effect. It turns out that substantially discriminative and generalizable speaker representations can be learnt in severely interfered conditions via our self-supervised training. The experimental results verify this seeming paradox. The learnt speaker embedding has superior discriminative power than a standard speaker verification method; meanwhile, Tune-In achieves remarkably better speech separation performances in terms of SI-SNRi and SDRi consistently in all test modes, and especially at lower memory and computational consumption, than state-of-the-art benchmark systems.
\end{abstract}

\section{Introduction}
The cocktail party effect \cite{narayan2009cortical,Samuel2015cocktail,getzmann2016switch, Getzmann2015neuro} is the phenomenon of the ability of a human brain to focus one's auditory attention to ``tune in" a single voice and ``tune out" the competing others. Although how humans solve the cocktail party problem \cite{cherry1953some} remains unknown,
remarkable progress on speech separation (SS) has been made thanks to recent advances of deep-learning models. Ground-breaking successful models include the high-dimensional embedding-based methods, originally proposed as the deep clustering network (DPCL) \cite{hershey2016deep}, and its extensions such as DANet \cite{chen2017deep} and DENet \cite{wang2018deep}. More recently, the performance has been improved by the time-domain audio separation networks (TasNet) \cite{luo2018tasnet} and the time convolutional networks (TCNs) based Conv-TasNet \cite{bai2018empirical, luo2019conv,lam2020mixup}. Until very recently, 
State-of-the-art (SOTA) performance record of TCNs have been further advanced on several benchmark datasets by the dual-path RNN (DPRNN) \cite{luo2019dual} and the dual-path gated RNN (DPGRNN) \cite{Eliya2020Voice}.

Modern deep-learning models also evolve rapidly for speaker identification and verification tasks \cite{David2018xvector, liu2018sv, mirco2018sincnet}. A separately trained speaker identification network is adopted in \cite{Eliya2020Voice} to prevent channel swap for SS. Other emerging studies \cite{zeg2020wavesplit, bo2018spkspIJCAI, bo2018spkspAAAI, bo2019spksp, wang2019Voicefilter, zi2020tastas} suggest bridging the speaker identification and the SS tasks in a unified network. Especially, Wavesplit in \cite{zeg2020wavesplit} jointly trains the speaker stacks and the separation stack to outperform DPRNN and DPGRNN. Similar to \cite{lucas2018deepatt}, it is based on K-means clustering, which results in significantly higher computational cost \cite{Pariente2020Asteroid}. According to its reported parameter setting, the estimated model size is approximately $42.5$M, about $18$ times larger than ours. To the best of our knowledge, higher SOTA scores than DPRNN's were all achieved by remarkably larger models. For another example, \cite{Eliya2020Voice} proposed a model of $7.5$M parameters, about $3$ times larger than ours; still, our system outperforms theirs in terms of SI-SNRi \cite{le2019sdr} and SDRi \cite{bo2018spkspIJCAI} on a benchmark dataset. We believe that it is only meaningful when the scores are compared under a fair constraint on the model sizes, as what we followed carefully and convincingly throughout our experiments.

A neurobiology study in \cite{Samuel2015cocktail} examines the modulation of neural activity associated with the interference properties. It demonstrates that a competing speech is processed predominantly within the same pathway as the target speech in the left hemisphere but is not treated equivalently within that stream.
Another inspiring research \cite{getzmann2016switch} suggests an acceleration in attention switching and context updating when people have semantically cued changes (e.g., a word) in target speaker settings than uncued changes. These recent neurobiology studies \cite{Samuel2015cocktail,getzmann2016switch,Getzmann2015neuro,narayan2009cortical} together suggest that, when listening with interference, the neural processes are more sophisticated than the simplified presumption in the previous attention-based separation \cite{bo2018spkspIJCAI}.

Inspired by the above, we propose a novel attention network that entails fundamentally different bottom-up and top-down mechanisms in stark contrast to the prior arts. Our contribution lies mainly in four folds:
\begin{itemize}[leftmargin=*]
\item a ``globally attentive, locally recurrent" (GALR) architecture (Sec. \ref{sec:GALR}) that breaks the memory and computation bottlenecks of self-attention and permits its usage over very long sequences (e.g., raw waveforms of speech mixtures);
\item a novel cross- and dual-attention mechanism (Sec. \ref{sec:crsndual}), which enables the bottom-up and top-down casting of signal-level and speaker-level information in parallel;
\item a theoretically grounded contrastive estimation loss, namely the Tune-InCE loss (Sec. \ref{sec:losses}) that can significantly enhance the robustness and generalizability of the learnt speaker embedding;
\item a new training and inference framework (Sec. \ref{sec:train_infer}) that efficiently solves the speaker embedding and speech separation problems and achieves reciprocity between the two tasks. In particular, our system can perform speech verification (SV) directly on overlapped mixtures with a performance comparable to using clean utterances, as discussed in Sec. \ref{sec:results}.
\end{itemize}

Conventional SV approaches require a complicated pipeline including first a speech activity detection module to remove noisy or silent parts, followed by a segmentation module, of which the output short segments are then grouped into corresponding speakers by a clustering module. Besides, to deal with interfered signals, an overlap detector and classifier are needed to remove the overlapped parts. Recent efforts have been made \cite{huang2020sd} to simplify such complicated pipelines into one stage. Still, performances are severely hurt in highly interfered scenarios \cite{yusuke2019sd}. In contrast, our framework has a quite different principle so that it can be relieved from the long pipeline. We call it ``Tune-In", abbreviated for ``Training under negative environments with Interference" to mimic the cocktail party effect, where a cross-attention mechanism automatically retrieves and summarizes the most salient, relevant, and reliable speaker embedding from a non-purified embedding pool, while discarding the noisy, vacant, unreliable, or non-relevant parts. Therefore, the summarized embedding is expected to capture the speaker's characteristics that are most salient and relevant to the current observation.

``Tune-In" together with the ``Contrastive Estimation" loss (that is the ``Tune-InCE" loss) views the commonplace masking routine in SS from a novel perspective, as a self-supervised technique for robust representation learning, like the masking in BERT \cite{Jacob2019Bert}. As far as we know, our model is also the smallest among all SOTA SS models, with an $11.5\%$ model size reduction relative to the previously smallest model --- DPRNN. Our model also achieved the SOTA SI-SNRi and SDRi scores at a much lower cost of up to $62.9\%$ less run-time memory and $66.4\%$ fewer computational operations according to results in Sec. \ref{sec:results}, where more insights are discussed in the ablation studies. We also provide theoretical and empirical analysis in the Appendix.



\begin{figure*}[t!]
 \centering
    \includegraphics[width=\linewidth]{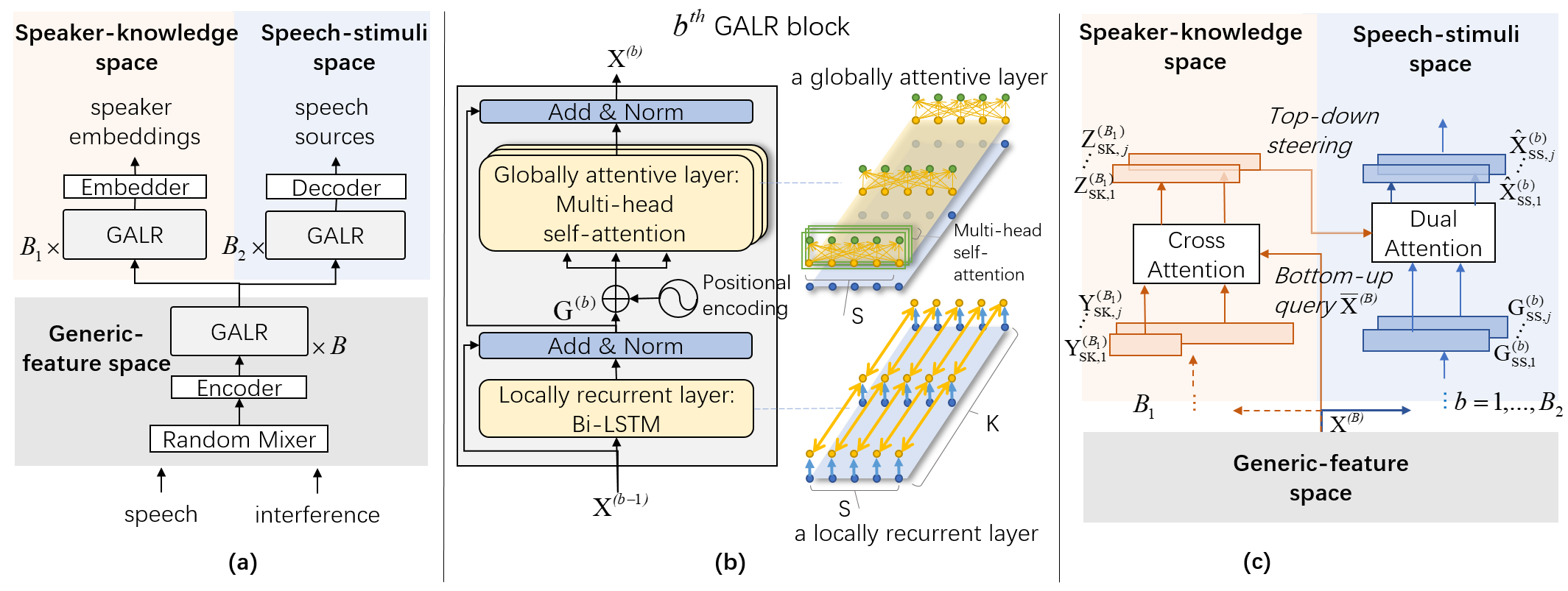}
\caption{(a) An overview of the proposed Tune-In framework, comprised of three spaces: a generic-feature space, a speaker-knowledge space, and a speech-stimuli space. (b) the structure inside a Globally Attentive, Locally Recurrent (GALR) block. (c) the bottom-up and top-down information casting between the two spaces via a cross- and a dual- attention mechanism.}
\label{fig:archi}
\end{figure*}
\section{Tune-In Framework}
\label{sec:tunein}
\subsection{Overview}
\label{sec:overview}
As shown in Fig. \ref{fig:archi} (a), the goal of the ``Speech-stimuli space" is to separate the $C$ sources from an input mixture $\mathbf{x}=\sum_{c=1}^{C}\mathbf{c}_i\cdot\mathbf{s}_c$, where $\mathbf{c}_i$ is a scaling factor to generate various signal-to-interference ratio (SIR) conditions (e.g., randomly from 0 to 5dB), and the sources are by $C$ speakers from the overall
$N$-speaker training corpus; the goal of the ``Speaker-knowledge space" is to extract speaker embedding to discriminate the $N$ speakers; the goal of the ``generic-feature space" is to extract shared generic features with discriminative and separable power among speakers and sources for both down-stream tasks.

The Encoder and the Decoder are the same as those in \cite{luo2019conv, luo2019dual} for comparison purposes, and we ask our readers to refer to these prior works for the detailed structure. The Encoder transforms the mixture $\mathbf{x}$ of the $C$ waveform sources into a sequential input $\mathbf{\Tilde{X}}\in \mathbb{R}^{D\times I}$, where $D$ is the feature dimension, and $I$ is the time-sequence length. We split $\mathbf{\Tilde{X}}$ into $S$ half-overlapping segments each of length $K$, and pack them into a 3-D tensor $\mathbf{X} \in \mathbb{R}^{D\times S \times K}$, and then pass it to a stack of 
blocks. These blocks share the novel structure namely the GALR structure. These GALR blocks stack together to accomplish tasks in all the spaces: 1) $B$ blocks in the generic-feature space constitutes the basis part that works on generic feature extraction. 2) $B_1$ and $B_2$ GALR blocks are fine-tuned towards different downstream tasks of speaker identification and speech separation, respectively.

\subsection{GALR Block}
\label{sec:GALR}
As shown in Fig. \ref{fig:archi} (b), we propose the GALR block structure to break the memory and computational bottleneck of self-attention mechanisms on tasks dealing with very long sequences. Each GALR block receives a 3-D tensor $\mathbf{X}^{(b-1)} \in \mathbb{R}^{D\times S \times K}$ and outputs $\mathbf{X}^{(b)}$ of the same shape, for $b = 1, ..., B$, where $B$ denotes the number of blocks in the generic-feature space. To avoid confusion, we denote $\mathbf{X}^{(b)}_\text{SK}$ as the output of GALR blocks in the speaker-knowledge (SK) space with $b = 1, ..., B_1$, and correspondingly $\mathbf{X}^{(b)}_\text{SS}$ for the speech-stimuli (SS) space with $b = 1, ..., B_2$. Each GALR block consists of two primary layers: a locally recurrent (LR) layer that performs intra-segment processing, and a globally attentive (GA) layer that performs inter-segment processing. 

\subsubsection{Locally Recurrent Layer}
\label{sec:lr}
Similar to DPRNN, we use RNNs for modeling local dependencies at the lower intra-segment context level, e.g., signal continuity, signal structure, etc, which are inherently important for waveform reconstruction. Consequently, a BiLSTM layer is responsible for modeling short-term dependencies within each segment and outputs  $\mathbf{L}^{(b)} \in \mathbb{R}^{D\times S \times K}$, where $
    \mathbf{L}_{s}^{(b)} = \text{LN}\left(\text{Lin}^{(b)}(\text{BiLSTM}^{(b)}(\mathbf{X}^{(b-1)}_{:,s,:}))\right)+\mathbf{X}^{(b-1)}_{:, s, :}$, for $s= 1,..., S$ refers to the local sequence within the $s^{th}$ segment, and $\text{BiLSTM}(\cdot)$, $\text{Lin}(\cdot)$, $\text{LN}(\cdot)$ denote a Bi-directional LSTM layer, a linear projection, and a layer normalization module \cite{ba2016layer}, respectively. 
\subsubsection{Globally Attentive Layer}
\label{sec:ga}
RNNs are far more sensitive to the nearby elements than to the distant ones. For example, an LSTM-based language modeling (LM) is capable of using about $200$ tokens of context on average, as revealed by \cite{khandelwal2018sharp}, but sharply distinguishes nearby context (recent $50$ tokens) from the distant history. As for speech, high temporal correlation, continuous acoustic-signal structure generally exist more often in intra-segment sequences ($0.25$ms- $2$ms granularity in our settings) than in inter-segment sequences ($64$ms- $512$ms granularity). This suggests that RNNs are potentially more suited for intra-segment modeling than for inter-segment modeling. Moreover, \cite{ravanelli2018light} discovered that RNNs reset the stored memory to avoid bias towards an ``unrelated'' history.

One strength of attention mechanisms over RNNs lies in a fully connected sequence processing strategy \cite{vaswani2017attention}, where every element is connected to other elements in a sequence via a direct path without any recursively processing, memorizing reset, or update mechanisms like RNNs. However, a conventional self-attention layer remains impractical to be applied in speech separation tasks, as very long speech signal sequences are involved, thus run-time memory consumption becomes a bottleneck. In terms of computational complexity, self-attention layers are also slower than recurrent layers when the sequence length is larger than the representation dimensionality, like most in speech separation.

Our proposed GALR block resorts to the self-attention mechanism \cite{vaswani2017attention} for inter-segment computation to model context-aware global dependencies. It revises DPRNN to better address long-range context modeling for audio separation, leading to a lower-cost, higher-performance structure. Since the hyperparameter $K$ in GALR controls the granularity of the local information modeled by the locally recurrent layer, the overall sequence length $I$ that the attention model deals with is reduced by a factor of $K$, which considerably relieves the run-time memory consumption. Note that self-attention has been used in spectrogram-based source enhancement by \cite{amir2019wildmix,yuma2020speech,jing2020DPTrans} recently, but never orthogonally across RNN-summarized frames, as in our proposed method. 
To validate that self-attention and RNN are better summarizers across segments and within the segment, respectively, we performed an ablation study in Appendix \textbf{5.3}. 

Moreover, the efficiency can be further improved by reducing the times of running the global attention mechanism, which is originally $K$ times over the $S$-length sequences in $\mathbf{L}^{(b)} \in \mathbb{R}^{D\times S \times K}$. We first pool the output of the locally recurrent model $\mathbf{L}^{(b)}$ through a $1\times 1$ convolution $\text{Conv2D}(\cdot)$ to map the $K$-dimension feature into $Q$ dimensions with $Q \ll K$. Then a layer normalization $\text{LN}_\text{D}(\cdot)$ \cite{ba2016layer} is applied along the feature dimension $D$. Subsequently, a positional embedding vector, $\mathbf{P}$, is added to the features:
$\mathbf{G}^{(b)} = \text{LN}_\text{D}\left(\text{Conv2D}(\mathbf{L}^{(b)})\right)+\mathbf{P}$, where $\mathbf{G}^{(b)} \in \mathbb{R}^{D\times S \times Q}$. Thus, the computation is reduced by a proportion of $(K-Q)/K$.
Then, we apply the multi-head self-attention $\text{Attn}(\cdot)$ to $\mathbf{G}^{(b)}$ following the encoder in Transformer \cite{vaswani2017attention}, of which the detailed equations are omitted here for their widely known operations. The output $\mathbf{\hat{X}}^{(b)} = \text{Attn}(\mathbf{G}^{(b)}) \in \mathbb{R}^{D\times S \times Q}$ goes through another 2D $1\times 1$ convolution to inversely map the $Q$ dimensions back to $K$ dimensions, yielding the final output of a GALR block $\mathbf{X}^{(b)} \in \mathbb{R}^{D\times S \times K}$. Likewise, each following block continues on the local-to-global and global-to-local (also fine-to-coarse and coarse-to-fine across the time axis) interactions.
\subsection{Self-supervised Learning with Cross- and Dual- Attention Mechanism}
\label{sec:crsndual}
On top of the generic-feature space, the following GALR blocks are tied to each downstream representation space. 
\subsubsection{Speech-stimuli Space}
\label{sec:spspace}
We want to learn deep speech representations for separating source signals in a mixture. The output of the last GALR block in this space, $\mathbf{X}^{(B_2)}_\text{SS}$, is passed to the Decoder that is the same as in DPRNN \cite{luo2019dual} \footnote{https://github.com/ShiZiqiang/dual-path-RNNs-DPRNNs-based-speech-separation}.

\subsubsection{Speaker-knowledge Space}
\label{sec:spkspace}
The goal here is to learn speaker representations that are easily separable and discriminative for identifying different speakers, including new unseen ones.
We pass the output of the last GALR block $\mathbf{X}^{(B_1)}_\text{SK} \in \mathbb{R}^{D\times {S_j} \times K}$ to the Embedder, where it is projected along the feature dimension from $D$ dimensions to $D'= C\times D$ dimensions. The resultant $(C\times D)\times {S_j} \times K$ dimensions are then averaged over the $K$ intra-segment positions and lead to $C\times (D\times {S_j})$ dimensional outputs, which are considered as $C$ speaker features $\mathbf{Y}^{(B_1)}_{\text{SK}, j}\in \mathbb{R}^{D\times S_j}, j=1,...,C$, where $S_j$ denotes the number of segments in $\mathbf{Y}^{(B_1)}_{\text{SK},j}$. 

\subsubsection{Cross Attention}
\label{sec:crossAtt}
As shown in Fig.\ref{fig:archi} (c), we design a cross- and dual- attention mechanism, where the two spaces interact only via queries. The bottom-up queries from the speech-stimuli space retrieve the most relevant information from the speaker-knowledge space and filter out non-salient, noisy, or redundant information, so that our system is spared from the long and complex pre-processing steps (VAD, segmentation, overlap cutting, etc.) to "purify" the speaker embedding. 
Inspired by the cross-stitch \cite{ishan2016cross} and the cross-attention methods \cite{, Hudson2018ComposAttn} used between text and image (knowledge base), we are the first to use a cross-attention strategy between speaker embedding and speech signal representations. $\text{Query}(\cdot)$, $\text{Key}(\cdot)$, and $\text{Value}(\cdot)$ denote linear transformation functions, where the corresponding input vectors are linearly projected along the feature dimension ($\mathbb{R}^{D}\mapsto \mathbb{R}^{D}$) into query, key, and value vectors, respectively. Our detailed implementation of these functions is the same as \cite{matt2018self} with code.\footnote{http://msperber.com/research/self-att}

The bottom-up query $\overbar{\textbf{X}}^{(B)}\in\mathbb{R}^{D\times S}$ is generated by averaging the generic-feature space's output $\textbf{X}^{(B)} \in \mathbb{R}^{D\times S \times K}$ over the $K$ intra-segment positions. Next, we cast the query vector onto the speaker-knowledge space via the cross-attention approach. Specifically, this is achieved by computing the inner product between $\text{Query}(\overbar{\textbf{X}}^{(B)})\in\mathbb{R}^{D\times S}$ and $\text{Key}(\mathbf{Y}^{(B_1)}_{\text{SK},j})\in\mathbb{R}^{D\times S_j}$:
\begin{equation}
\label{eq:a1}
   \mathbf{a}_{\textbf{crs},j}=\text{softmax}(\text{Query}(\overbar{\textbf{X}}^{(B)})^{\top}\cdot\text{Key}(\mathbf{Y}^{(B_1)}_{\text{SK},j})).
\end{equation}
yielding a cross-attention distribution $\mathbf{a}_{\textbf{crs},j}\in\mathbb{R}^{S\times S_j}$.

Finally, we compute the sum of $\text{Value}(\mathbf{Y}^{(B_1)}_{\text{SK},j})\in\mathbb{R}^{D\times S_j}$ weighted by the cross-attention distribution $\mathbf{a}_{\textbf{crs},j}$, and then average over $S$ segments, to produce for speaker $j$ a vector $\mathbf{Z}^{(B_1)}_{\text{SK},j}\in \mathbb{R}^{D}$ , which we call a steering vector in contrast to the following ``autopilot" mode that does not rely this vector:
\begin{equation}
\label{eq:a2}
    \mathbf{Z}^{(B_1)}_{\text{SK},j}=\frac{1}{S}\sum_{S}\sum_{S_j}\mathbf{a}_{\textbf{crs},j}\cdot\text{Value}(\mathbf{Y}^{(B_1)}_{\text{SK},j})^\top.
\end{equation}

For notation simplicity, we denote $\mathbf{Z}^{(B_1)}_{\text{SK},j}$ as $\mathbf{Z}_{j}$ from now on. It is expected to capture the speaker's characteristics that are most salient and relevant to the currently observed content. Note that the generalization power of the steering vector $\mathbf{Z}_{j}$ is vital for maintaining the separation performance over unseen speakers. Therefore, based on Eq. (\ref{eq:a2}), we investigated two techniques to increase the robustness of the steering vectors: (1) apply embedding dropout $\text{Dropout}(0.1)$, and (2) add element-wise Gaussian embedding noise $\mathcal{N}(0, 0.1)$, to the steering vector $\mathbf{Z}_{j}$ only at training time. We study the ablations of each part of these techniques in Sec. \ref{sec:ablation}.

\subsubsection{Dual Attention}
\label{sec:dual}
A top-down steering is raised from the speaker-knowledge space by projecting the steering vector onto the speech-stimuli space via the dual-attention approach. It mimics the neural process that uses specific object representations in auditory memory to enhance the perceptual precision during top-down attention \cite{obleser2015selective}. This is simulated here by firstly passing $\mathbf{Z}_{j}$ through two linear mappings $\text{r}(\cdot)$ and $\text{h}(\cdot)$ to modulate the original self-attention input $\textbf{G}^{(b)}_{\text{SS},j}$ and to generate keys and values that capture the top-down speaker information. Then we compute the dual-attention distribution $\mathbf{a}_{\textbf{dual},j}$ as below:
\begin{equation}
\label{eq:a3}
\begin{aligned} 
    \mathbf{a}_{\textbf{dual},j}=\text{softmax}(&\text{Query}(\textbf{G}^{(b)}_{\text{SS},j})^{\top}\cdot \\
    &\text{Key}(\text{r}(\mathbf{Z}_{j})\odot\textbf{G}^{(b)}_{\text{SS},j}+\text{h}(\mathbf{Z}_{j})))
\end{aligned}
\end{equation}
where $\textbf{G}^{(b)}_{\text{SS},j}$ is the input of each globally attentive layer for $b=0,...,B_2$ in the speech-stimuli space. Note that $\textbf{G}^{(0)}_{\text{SS},j}$ is equal for all $j=1,...,C$, as they share the value computed from the last output $\textbf{X}^{(B)}$ from the generic-feature space.

Finally, we sum the values weighted by $\mathbf{a}_{\textbf{dual},j}$ to produce the modified self-attention output:
\begin{equation}
\label{eq:a4}
    \hat{\textbf{X}}^{(b)}_{\text{SS},j}=\sum_{S}\mathbf{a}_{\textbf{dual},j}\cdot\text{Value}(\text{r}(\mathbf{Z}_{j})\odot\textbf{G}^{(b)}_{\text{SS},j}+\text{h}(\mathbf{Z}_{j}))^\top
\end{equation}

Notably, Eq. (\ref{eq:a3}-\ref{eq:a4}) together can be seen as a replacement of the original self-attention operation in each globally attentive layer in the speech-stimuli space. In Appendix \textbf{5.4}, we empirically validate and analyze the effectiveness of dual attention by comparing it to another successful approach --- feature-wise linear modulation (FiLM) \cite{Ethan2017film}.

\subsection{Losses}
\label{sec:losses}
In the speech-stimuli space, we use the scale-invariant signal-to-noise ratio (SI-SNR) \cite{le2019sdr} loss, $\mathcal{L}_{\tiny\text{SI-SNR}}$, to reconstruct the source signals using an utterance-level permutation invariant training (u-PIT) method \cite{yu2017permutation}. However, losses such as SI-SNR expend heavy computation at reconstructing every detail of the signals, while often ignoring the global context. Instead, in the speaker-knowledge space, we want to learn robust representations that entail the underlying shared-speaker context among one speaker's different utterances of speech signals under different interfering signals. Therefore, we design more appropriate losses $(\mathcal{L}_{\tiny\text{InCE}} + \mathcal{L}_{\tiny\text{reg}})$ below for extracting the underlying shared information between the extended sequence context. All in all, the joint loss is $\mathcal{L}_{\tiny\text{SI-SNR}}+\lambda(\mathcal{L}_{\tiny\text{InCE}}+\mathcal{L}_{\tiny\text{reg}})$, where $\lambda$ is a weighting factor. Appendix \textbf{5.7} elaborates the speaker-embedding-based permutation computation for training speedup. 
\subsubsection{Self-supervised Loss}
\label{sec:ince}
$\mathcal{L}_{\tiny\text{InCE}}$ is our proposed self-supervised loss called Tune-InCE loss defined as follows:
\begin{align}
\label{eq:Tune-InCE}
    &\mathcal{L}_{\tiny\text{InCE}} = -
    \mathbb{E}_{\mathcal{D}}\left[\log\left({f(\mathbf{Z}_{j}, \textbf{E}_{i_j})}/{\sum_{i=1}^{N}f(\mathbf{Z}_j, \textbf{E}_i)}\right)\right],\\
    &f(\mathbf{Z}_{j}, \textbf{E}_{i})= \exp\left(-\alpha\lVert\mathbf{Z}_{j}- \textbf{E}_{i}\rVert_{2}^{2}\right), \label{eq:f}
\end{align}
where $\textbf{E}_i$ is the speaker embedding vector defined for all $i=1,...,N$ speakers in the training set, $i_j$ indicates the corresponding speaker index for speaker $j=1,...,C$ in the current mixture, $\alpha>0$ is a learnable scaler, and $\mathbb{E}_{\mathcal{D}}\left[\cdot\right]$ denotes the expectation over the training set $\mathcal{D}$ containing all input mixture utterances and corresponding speaker IDs. 

In Appendix \textbf{5.1}, we prove that our proposed form of Eq. (\ref{eq:f}) about $f(\mathbf{Z}_{j}, \textbf{E}_{i})$  corresponds to treating each speaker embedding $\textbf{E}_{i_j}$ as a cluster centroid of different steering vectors $\mathbf{Z}_{j}$ of utterances generated by the same speaker $i_j$ while $\alpha > 0$ controls the cluster size. $f(\mathbf{Z}_{j}, \textbf{E}_{i})$ can also be regarded as a regularized variant of the contrastive predictive coding (CPC) loss \cite{nce18, Yoshua2019self1, Yoshua2019self2, Yoshua2019self3, Yoshua2019self4}. We also prove that by minimizing $\mathcal{L}_{\tiny\text{InCE}}$ we maximize a lower bound of the mutual information between the steering vector $\mathbf{Z}_{j}$ and the speaker vector $\textbf{E}_{i_j}$, which encourages learning separable inter-speaker utterance embeddings and compact intra-speaker utterance embeddings. During training, whenever a steering vector of source $j$ by the speaker $i_j$ is generated, the corresponding speaker embedding vector is updated using exponential moving average (EMA):
$\Delta\textbf{E}_{i_j}=\varepsilon(\textbf{Z}_{j}-\textbf{E}_{i_j})$, where $\varepsilon\in[0,1]$ is a small hyper-parameter scalar that controls the update step size.


$\mathcal{L}_\text{reg}$ is a regularization loss that avoids collapsing to a trivial solution of all zeros:
\begin{equation}
\label{eq:loss-reg}
\mathcal{L}_\text{reg}=-\frac{1}{\gamma C}\sum_{j=1}^C\min_{i\in\{1,...,N\}\backslash \{i_j\}}
\text{log}\lVert\mathbf{E}_{i_j}-\mathbf{E}_i\rVert_1,
\end{equation}
where $\gamma > 0$ is a weighting factor. See Table \ref{tab:ablation} for $\mathcal{L}_\text{reg}$'s ablation study.

\subsection{Training and Inference Modes}
\label{sec:train_infer}
We train and test the proposed Tune-In system sequentially in different modes.

1) ``Autopilot" mode: performing a standalone speech separation without any speaker knowledge, when the training graph only traverses the generic-feature space and the speech-stimuli space, i.e., the forward and backward path goes only along the blue lines as shown in Fig. \ref{fig:archi} (c), and only the SI-SNR loss is used. Since in ``Autopilot" mode there is no top-down steering information from the speaker-knowledge space, the outputs of $r(\cdot)$ and $h(\cdot)$ in Eq. (\ref{eq:a3}-\ref{eq:a4}) are an all-one vector and an all-zero vector, respectively.

2) ``Online" mode: performing concurrent speaker-knowledge extraction and speech separation, based on the same online input signals, i.e.,  ${S_j}= S$. Now the training forward path goes along both the blue and the brick-red lines in Fig. \ref{fig:archi} (c). In this mode, the computation in the generic-feature space is shared by both down-streaming spaces.

3) ``Offline" mode: performing offline speaker-knowledge extraction that utilizes a given enrollment (i.e., ${S_j}\neq{S}$) with the target speaker known a priori. The enrollment was not restricted to be a mixture with another random interfering speaker (with no prior knowledge) or a clean utterance. Since the speaker embedding could been pre-computed offline, this part of the inference time is saved. In Fig. \ref{fig:archi} (c), we depict this property with different segment lengths and dotted lines connecting the generic-feature space and the speaker-knowledge space. Unlike traditional SV systems, we can extract reliable speaker embedding from not only clean utterances but also challenging overlapped mixtures. This is significant merit particularly important to realistic industrial deployment.

Until now, it is free to collect speaker enrollments from unconstrained various environments. Our proposed bottom-up query approach makes the model only attend to the most relevant parts of variable-length sequences of either online signals (analog to recent short-term memory) or pre-collected enrollments (analog to long-term persistent memory and experience). The top-down query then integrates the retrieved information to iteratively compute the speech representation through the stack of GALR blocks in the speech-stimuli space. All the above attentive operations are directly inferred from the input data without resorting to supervision.

In both ``online" and ``offline" modes, the model components in the speaker-knowledge space and the speech-stimuli space are trained cooperatively by casting information towards each other via the cross- and dual-attention mechanism. Our joint training method is very different from existing approaches that generally fuse the speaker embeddings and speech features into the same vector space through linear combinations, multiplication, or concatenation \cite{jixuan2020spk, Ethan2017film, zeg2020wavesplit}. Specifically, we facilitate the communication of the two spaces only through soft-attention mechanisms. Consequently, the interaction between the two spaces is mediated through probability distributions only. The casting of speech features onto speaker embedding restricts the space of the valid top-down steering operations by anchoring the steering vectors $\mathbf{Z}_{j}$ back in the original speaker-knowledge space. This serves as a form of regularization that restricts $\mathbf{Z}_{j}$ to be a compact vector in the unique speaker-knowledge space, which enables the self-supervised techniques as discussed in Sec. \ref{sec:losses}. 

Moreover, our network achieves better transparency and interpretability than the existing approaches. Unlike the black-box network in prior works, we can easily interpret the cross attention and the consequent top-down steering process. In Appendix \textbf{5.5}, we illustrate an intriguing phenomenon that selective bottom-up cross attention similar to humans' behavior could be automatically learnt by our network.
\subsubsection{External Speaker Augmentation}
\label{sec:spk-aug}
We find that data augmentation with external speaker datasets can enhance the generalizability of learnt speaker embedding over unseen speakers during a test. It is particularly useful when the number of speakers is small in the original training dataset. For example, we used Librispeech for the speaker augmentation for WSJ0-2mix in Table \ref{tab:ablation}.
Our framework is flexible to utilize external speaker datasets to continue training the blocks in the speaker-knowledge space, where the forward and backward path goes only along the brick-red lines as shown in Fig. \ref{fig:archi}, and only the speaker loss $\mathcal{L}_{\tiny\text{InCE}} + \mathcal{L}_{\tiny\text{reg}}$ is used.
After that, the blocks in the speech-stimuli space are fine-tuned using the original training set with the extracted $\mathbf{Z}_{j}$ by the more generalized speaker-knowledge component. Table \ref{tab:ablation} provides an ablation study regarding this technique.

\section{Evaluation and Analysis}
\label{sec:3}
\subsection{Datasets and Model Setup}
\label{sec:dataset_model}
We describe the datasets and model setup in brief as below. Additional information needed to reproduce our experiments is stated in detail in Appendix \textbf{5.2}. 
\subsubsection{Datasets}
\label{sec:dataset}
We used a benchmark 8kHz dataset WSJ0-2mix \cite{hershey2016deep} for comparison with state-of-the-art source separation systems. Moreover, to evaluate performance on separating speech from music interference, we created another dataset WSJ0-music by mixing the speech corpus of WSJ0-2mix with music clips. We also used a large-scale publicly available 16kHz benchmark dataset Librispeech \cite{panayotov2015librispeech}. 

During training, each utterance is online mixed (masked) with an interfering utterance from another randomly chosen speaker. For testing, we pre-generated a fixed set of interfered utterances using the same SNR conditions as in training.

\begin{table*}[ht!]
\centering
\begin{tabular}{c|c|c|cc|cc}
\bf Dataset & \bf Model & \bf  Parameters & \bf  Memory & \bf  GFLOPs & \bf SI-SNRi &\bf SDRi  \\
\hline
\multirow{7}{*}{WSJ0-2mix}
& {TDAA \cite{bo2018spkspIJCAI}} & 14.8M & - & -& - & 12.6 \\
& {Conv-TasNet} & 8.8M &-&-& 15.5 & 15.9 \\
& {$^\ddagger$DPRNN} & 2.6M &1,970MB&84.6& 18.8 & 19.1 \\
\cline{2-7}
& Tune-In Autopilot & \bf 2.3M & \bf 730MB & \bf 28.4 & \bf 20.3 & \bf 20.5 \\
& Tune-In Offline & \bf 2.4M & \bf 799MB & \bf 28.9 & \bf 20.4 & \bf 20.6\\
& Tune-In Online & 3.2M & \bf 1,309MB & \bf 37.3& \bf 20.8$^\dag$ & \bf 21.0$^\dag$\\
\hline\hline
\multirow{3}{*}{WSJ0-music}
& {$^\ddagger$DPRNN} & 2.6M & 231MB& 10.7 & 14.5 & 14.8 \\
\cline{2-7}
& Tune-In Autopilot & \bf 2.3M &\bf 186MB& \bf 8.3 & \bf15.9 & \bf 16.2\\
& Tune-In Online & 3.2M &332MB& 10.7 &  \bf 15.9 & \bf 16.2\\
\hline\hline
\multirow{4}{*}{LibriSpeech}
& {$^\ddagger$DPRNN} & 2.6M &1,152MB&40.5& 12.0 & 12.5 \\
\cline{2-7}
& Tune-In Autopilot & \bf  2.3M &\bf 1,013MB&\bf 31.1&  \bf 12.2 & \bf 12.7 \\
& Tune-In Offline & \bf 2.4M & \bf 1,112MB& \bf 31.7& \bf 12.6 & \bf 13.1 \\
& Tune-In Online & 3.2M & 1,912MB &\bf 39.9& \bf 13.0 & \bf 13.6 \\
\end{tabular}
\caption{
Performance comparison with conventional separation systems. Models with $\ddagger$ were reproduced in our own implementation for performance and runtime costs evaluation. Memory and GFLOPs were calculated on a 1s audio input.
}
\label{tab:sisnr}
\end{table*}

\subsubsection{Model Setup}
\label{sec:modelsetup}
The encoder and decoder structure, as well as the model's hyper-parameter settings, were directly inherited from DPRNN's setup \cite{luo2019dual} for comparison purposes. Note that no model hyper-parameter has been fine-tuned towards our proposed structure, otherwise more improvement could be reasonably expected for ours. Tune-In loss hyper-parameters were set empirically.

\subsection{Results and Discussion}
\label{sec:results}
\subsubsection{Speaker Verification Performance}
This experiment investigated a downstream SV task to verify the discriminative and generalization property of our learnt representations. Note that we did not aim to achieve a SOTA SV score, but to examine the effectiveness of representation learning. An advanced SincNet-based SV model \cite{Yoshua2018sv} was considered a qualified reference system, for its reproducible high performance and speaker-embedding based network.\footnote{https://github.com/mravanelli/SincNet.}
One of the reference models, namely ``$\text{SincNet-clean}$", was trained on the clean Librispeech training set like conventional SV settings; the other reference model, namely ``$\text{SincNet-masked}$", was trained on the online-masked (mixed) Librispeech training set same as our proposed system. 
\label{sec:svperform}
\begin{figure}[t!]
 \centering
    \includegraphics[width=0.46\textwidth]{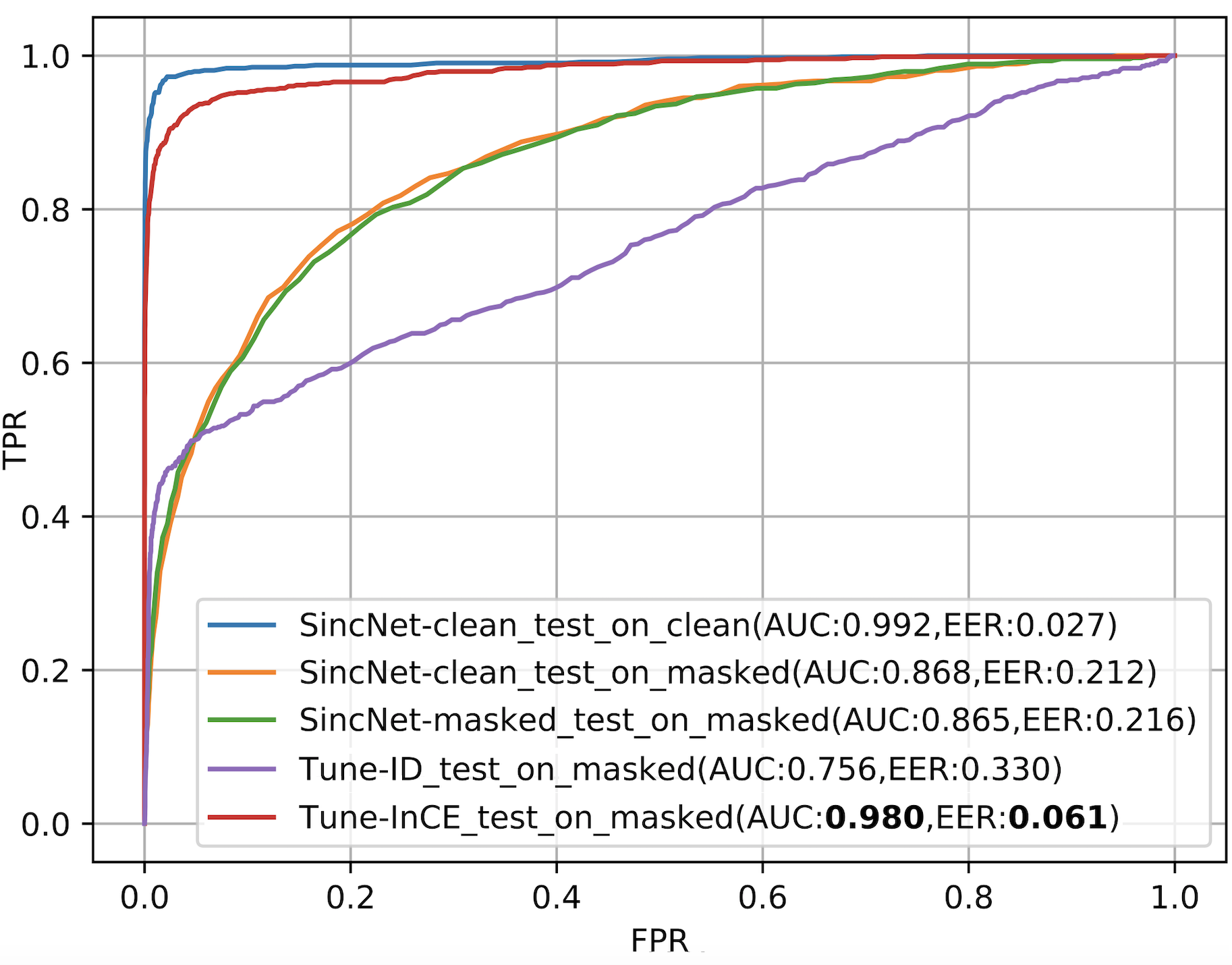}
\caption{ROC curves of speaker verification by different models.}
\label{fig:roc}
\end{figure}
Fig. \ref{fig:roc} plots the true positive rate (TPR) against the false positive rate (FPR), and gives the Receiving Operating Characteristic (ROC) curves by different models. The blue line indicates the standard performance by the SincNet SV system  on clean test data (``\text{SincNet-clean\_test\_on\_clean}") and can serve as an upper-bar reference. The yellow and green lines indicate that SincNet, either trained with clean (``\text{SincNet-clean\_test\_on\_masked}") or masked data (``\text{SincNet-masked\_test\_on\_masked}"), can not handle masked test data well. The red line denoting the ``\text{Tune-InCE\_test\_on\_masked}" was by our proposed self-supervised model on the corrupted test data in the adverse conditions with $0$-$5$dB SIR. Surprisingly, it achieves an AUC as high as $0.980$ and an EER of $0.061$, and its ROC line is approaching the ``\text{SincNet-clean\_test\_on\_clean}", i.e., the performance on clean data by the reference SV system.

We also trained our model using a fully supervised approach of learnable token embedding, which has been widely applied in both NLP \cite{Jacob2019Bert} and speech domain \cite{bo2018spkspIJCAI, zeg2020wavesplit}. The learnable embedding \cite{bo2018spkspIJCAI} maintained the target embedding of the inventory of the $N$ speaker IDs. We call the fully supervised model ``Tune-ID", where the token embeddings for speaker IDs replaced our EMA based speaker embedding $\textbf{E}_{i_j}$ in ``Tune-InCE", and the conventional token embedding loss \cite{bo2018spkspIJCAI}, replaced our proposed self-supervised Tune-InCE loss. All other training conditions for ``Tune-ID", such as Gaussian noise, dropout, online mixing strategy, etc., are the same as ``Tune-InCE". We can observe that the fully supervised model ``\text{Tune-ID\_test\_on\_masked}" performs poorly in generalization. To analyze the above performances, in Appendix \textbf{5.6}, we study the generalization ability of the deep representations learnt by our proposed self-supervised approach and compare it to the fully supervised approach. The experimental result suggests that SincNets, even the one trained with the masked, are fragile when the speech inputs are not clean enough. In contrast, the self-supervised learning prevented our model from learning a trivial task of only predicting speaker identity, but enforced the model to learn deep representations with essential discriminative power.
\subsubsection{Speech Separation Performance}
We then evaluated our proposed self-supervised Tune-In model on a speech separation task. As shown in Table \ref{tab:sisnr}, all systems were assessed in terms of SI-SNRi and SDRi. The ``Tune-In Online" model inferred the speaker embeddings and steering vectors online and performed speech separation simultaneously; thus the additional task brought the extra complexity and around $0.8$M model parameter size more than that of ``Tune-In Autopilot" and ``Tune-In Offline". The results with $^\dag$ indicate speaker augmentation was applied during training (see Table \ref{tab:ablation} for ablation study). As discussed in Sec.\ref{sec:train_infer}, ``Tune-In Offline" uses enrollments from target speakers. These enrollments were $16$s long and collected from challenging conditions as online, i.e., the same SIR range but with different interferences.

Our Tune-In model outperforms the DPRNN model by a very large margin in terms of SI-SNRi and SDRi, which is achieved at a lower cost of up to $62.9\%$ less run-time memory and $66.4\%$ fewer computational operations. Even compared to the most recent SOTA model (DPGRNNs) \cite{Eliya2020Voice} with parameters as large as $7.5$M, our Tune-In model still can maintain a significant advantage of $0.7$ in terms of SI-SNRi and SDRi. For WSJ0-music, we used a window length of 16 samples instead of 4 samples to speed up the training for both Tune-In and DPRNN systems. By comparing results on WSJ0-music by Tune-In Autopilot v.s. Online, we noticed that the speaker embedding was not as beneficial for the speech-music separation task. This is reasonable because speech-music separation tends to rely on other key factors, e.g., rhythm, spectral continuity, pitch, etc., rather than speaker embedding. Nevertheless, results on WSJ0-music demonstrate the consistent advantage of our proposed GALR structure over DPRNN by a large margin.
\subsubsection{Ablation Study}
\label{sec:ablation}
In Table \ref{tab:ablation}, we presented the test SI-SNRi for variants of our best-performing Tune-In Online system. Each variant removes a technique that is proposed for presumably better speaker embedding learning. The first two ablations experimented regarding the two different regularization techniques proposed after Eq. (\ref{eq:a2}). When the dropout replaces the element-wise Gaussian noise, we observed a notable degradation. This empirically suggests that adding Gaussian noise rather than dropout is effective for building a more robust speaker-knowledge-steered separation system. The removal of speaker augmentation also degrades the performance. Meanwhile, we found that the removal of regularization loss would bring down SI-SNRi. The ablation results together suggest that the combination of the proposed techniques is pivotal in achieving state-of-the-art performance.
\begin{table}[b!]
\centering
\begin{tabular}{l|c}
\bf Model & \bf SI-SNRi \\
\hline
Tune-In Online & 20.8 \\
\hline
\,\,w/ Embedding Dropout & 19.8\\
\,\,w/o Embedding Noise & 20.0\\
\,\,w/o Speaker Augmentation & 20.6\\
\,\,w/o Regularization Loss & 20.2\\
\end{tabular}
\caption{
Model ablations for our Tune-In Online model reporting results over the test set on WSJ0-2mix.
}
\label{tab:ablation}
\end{table}
\section{Conclusions}
\label{sec:4}
``Tune-In" attempts to practice the cocktail party effect in a new attentive way. We find it is critical to learn and maintain a speaker-knowledge and a speech-stimuli space separately. A novel cross- and dual-attention mechanism is proposed for information exchange between the two spaces, mimicking the bottom-up and top-down processes of a human neural system. The GALR structure as the building block in Tune-In breaks the memory and computation bottleneck of conventional self-attention layers. We reveal that substantially discriminative and generalizable representations can be learnt in severely interfered conditions via our self-supervised training. Tune-In outperforms a standard SV system and meanwhile achieves SOTA SS performances consistently in all test modes, particularly at a lower cost. Future work includes 1) a multi-channel extension as human's cocktail party effect works best as a binaural effect, 2) the effect revealed by prior neurobiology study 
about an acceleration in attention switching and context updating, and 3) Tune-In for speech recognition by replacing the speech-stimuli space with a speech-phoneme space.



\bibliography{mybib}

\section{Appendix}
\subsection{Theoretical Studies for The Tune-InCE Loss}
In the main paper, we introduce the Tune-InCE loss, which takes the following form:
\begin{align}
\label{eq:1.1}
    &\mathcal{L}_{\tiny\text{InCE}} = -
    \mathbb{E}_{\mathcal{D}}\left[\log\left({f(\mathbf{Z}_{j}, \textbf{E}_{i_j})}/{\sum_{i=1}^{N}f(\mathbf{Z}_j, \textbf{E}_i)}\right)\right],\\
\label{eq:1.2}
    &f(\mathbf{Z}_{j}, \textbf{E}_{i})= \exp\left(-\alpha\lVert\mathbf{Z}_{j}- \textbf{E}_{i}\rVert_{2}^{2}\right),
\end{align}
For the ease of understanding, we consider the steering vector $\mathbf{Z}_{j}$ as an embedding for the $j$-th separated signal (referred to as separation embedding), for $j=1,...,C$, given that each audio input is a mixture of $C$ sources. $\textbf{E}_{i}$ is the speaker embedding for speaker $i$ out of $N$ speakers. We only concern a multi-talker scenario, where the source $j$ can only be generated by one of the speakers indexed by $i_j$. We also assume that the training set $\mathcal{D}$ sufficiently defined the sample space of their joint distributions.
\par
First, we would like to study the relationship between Tune-InCE loss and mutual information:

\begin{definition}
\label{def:1}
Mutual information of the speaker embeddings and the separation embeddings is defined as
\begin{align}
\label{eq:2}
\mathcal{I}(\textbf{E}_{i_j};\mathbf{Z}_{j})=\mathbb{E}_{\mathcal{D}}\left[\log \frac{p(\textbf{E}_{i_j}, \mathbf{Z}_{j})}{p(\textbf{E}_{i_j}) p(\mathbf{Z}_{j})}\right].
\end{align}
\end{definition}

\par
Then, we make the following assumptions:

\begin{assumption}
\label{ass:1}
With a suitable mathematical form for function $f(\cdot)$, we can model a density ratio defined as
\begin{align}
\label{eq:3}
f(\mathbf{Z}_{j}, \textbf{E}_{i_j}) \propto \frac{p(\textbf{E}_{i_j}|\mathbf{Z}_{j})}{p(\textbf{E}_{i_j})}=\frac{p(\mathbf{Z}_{j}|\textbf{E}_{i_j})}{p(\mathbf{Z}_{j})}.
\end{align}
\end{assumption}

\begin{assumption}
\label{ass:2}
Considering the case of $i\neq i_j$, since the separation embedding $\mathbf{Z}_{j}$ does not belong to speaker $i$, $\textbf{E}_{i}$ should not be dependent on $\mathbf{Z}_{j}$. Therefore, it is sensible to assume
\begin{align}
\label{eq:4}
p(\textbf{E}_{i}|\mathbf{Z}_{j})=p(\textbf{E}_{i}),&\,\,\,\,\,\,\,\,\forall\,\,i\neq i_j.
\end{align}
\end{assumption}

\begin{table*}[ht!]
\centering
\caption{
Performance comparison on WSJ0-2mix by DPRNN versus GALR, each with different window length settings.
}
\label{tab:winlen}
\begin{tabular}{c|c|c|cc|cc}
\specialrule{.16em}{0em}{0em} 
\textbf{Architecture}  & \textbf{Parameters} & \textbf{\makecell{Window\\Length}} & \textbf{SI-SNRi}  & \textbf{SDRi} & \textbf{\makecell{Memory}} & \textbf{GFLOPs} \\
\hline
\multirow{4}{*}{DPRNN} & \multirow{4}{*}{2.6M} & 16 & 15.9 & 16.2 & 231 MiB & 10.7\\ 
& & 8  & 17.0 & 17.3 & 456 MiB & 22.2 \\ 
& & 4  & 17.9 & 18.1 &929 MiB & 42.3\\ 
& & 2  & 18.8 & 19.1 &1,970 MiB & 84.6\\ 
\hline
\multirow{4}{*}{GALR} 
  & \multirow{4}{*}{\textbf{2.3M}}& 16 & \textbf{17.0}  &  \textbf{17.3} & \textbf{186 MiB} &\textbf{8.3}\\
 & & 8 & \textbf{18.7} &  \textbf{18.9}&  \bf{363 MiB} &\bf{14.2} \\
& & 4 & \textbf{20.3} &  \textbf{20.5} &\bf{730 MiB} &\bf{28.4}\\
& & 2 & \textbf{19.7} &  \textbf{19.9} &\bf{1,490 MiB} &\bf{55.5}\\
\specialrule{.13em}{0em}{0em} 
\end{tabular}
\end{table*}

After making these assumptions, we can deduce the following claims:
\begin{claim}
\label{claim:1}
Minimizing the Tune-InCE loss results in maximizing mutual information between the speaker embeddings and the separation embeddings, since the Tune-InCE loss $\mathcal{L}_{\tiny\text{InCE}}$ serves as an upper bound of the negative mutual information $-\mathcal{I}(\textbf{E}_{i_j};\mathbf{Z}_{j})$.
\end{claim}
\begin{proof}
To prove this claim, we substitute Eq. (\ref{eq:3}) into Eq. (\ref{eq:1.1}) and obtain the following results:
\begin{align*}
    \mathcal{L}_{\tiny\text{InCE}} &= -
    \mathbb{E}_{\mathcal{D}}\left[\log\left(\frac{\frac{p(\textbf{E}_{i_j}|\mathbf{Z}_{j})}{p(\textbf{E}_{i_j})}}{\sum_{i=1}^{N}\frac{p(\textbf{E}_{i}|\mathbf{Z}_{j})}{p(\textbf{E}_{i})}}\right)\right]\\
    &=\mathbb{E}_{\mathcal{D}}\left[\log\left(1+\frac{\sum_{i\neq i_j}\frac{p(\textbf{E}_{i}|\mathbf{Z}_{j})}{p(\textbf{E}_{i})}}{\frac{p(\textbf{E}_{i_j}|\mathbf{Z}_{j})}{p(\textbf{E}_{i_j})}}\right)\right]\\
    &=\mathbb{E}_{\mathcal{D}}\left[\log\left(1+(N-1)\frac{p(\textbf{E}_{i_j})}{p(\textbf{E}_{i_j}|\mathbf{Z}_{j})}\right)\right]\\
    &\geq\mathbb{E}_{\mathcal{D}}\left[\log\left(N\frac{p(\textbf{E}_{i_j})}{p(\textbf{E}_{i_j}|\mathbf{Z}_{j})}\right)\right]\\
    &=\log N-\mathbb{E}_{\mathcal{D}}\left[\log\left(\frac{p(\textbf{E}_{i_j}|\mathbf{Z}_{j})}{p(\textbf{E}_{i_j})}\right)\right]\\
    &=\log N-\mathbb{E}_{\mathcal{D}}\left[\log\left(\frac{p(\textbf{E}_{i_j}, \mathbf{Z}_{j})}{p(\textbf{E}_{i_j})p(\mathbf{Z}_{j})}\right)\right]\\
    &=\log N-\mathcal{I}(\textbf{E}_{i_j};\mathbf{Z}_{j})
\end{align*}
This proves that $\mathcal{L}_{\tiny\text{InCE}}$ is an upper bound of $ -\mathcal{I}(\textbf{E}_{i_j};\mathbf{Z}_{j})$.
\end{proof}
Next, we study our proposed form of $f(\mathbf{Z}_{j}, \textbf{E}_{i_j})$ and its associated properties when minimizing the Tune-InCE loss.
\begin{claim}
\label{claim:2}
Applying our proposed form of $f(\mathbf{Z}_{j}, \textbf{E}_{i_j})$ to $\mathcal{L}_{\tiny\text{InCE}}$ corresponds to treating each speaker embedding $\textbf{E}_{i_j}$ as a cluster centroid (Gaussian mean) of different separation embedding vectors $\mathbf{Z}_{j}$ generated by the same speaker $i_j$ with a learnable parameter $\alpha > 0$ controlling the cluster size (Gaussian variance):
\begin{align}
\label{eq:6}
    p(\mathbf{Z}_{j}|\textbf{E}_{i_j})&=\mathcal{N}(\textbf{E}_{i_j}, (2\alpha)^{-1}\mathbf{I})\\
\label{eq:7}
    p(\mathbf{Z}_{j})&=\mathcal{N}(\mathbf{0}, (2\alpha)^{-1}\mathbf{I}).
\end{align}
\end{claim}
\begin{proof}
Considering a density ratio $\hat{f}$:
\begin{align*}
\hat{f}(\mathbf{Z}_{j}, \textbf{E}_{i_j}) &\propto \frac{\mathcal{N}(\textbf{E}_{i_j}, (2\alpha)^{-1}\mathbf{I})}{\mathcal{N}(\mathbf{0}, (2\alpha)^{-1}\mathbf{I})}\\
&=\frac{\exp\left(-1/2\cdot(2\alpha)\lVert\mathbf{Z}_{j}-\textbf{E}_{i_j}\rVert_2^2\right)}{\exp\left(-1/2\cdot(2\alpha)\lVert\mathbf{Z}_{j}\rVert_2^2\right)}\\
&=\exp\left(-\alpha\lVert\mathbf{Z}_{j}-\textbf{E}_{i_j}\rVert_2^2+\alpha\lVert\mathbf{Z}_{j}\rVert_2^2\right)
\end{align*}
Now, we evaluate $\mathcal{L}_{\tiny\text{InCE}}$ with the above-defined $\hat{f}(\mathbf{Z}_{j}, \textbf{E}_{i_j})$:

\begin{align*}
    \mathcal{L}_{\tiny\text{InCE}} &= -\mathbb{E}_{\mathcal{D}}\left[\log\left(\frac{\exp\left(-\alpha\lVert\mathbf{Z}_{j}-\textbf{E}_{i_j}\rVert_2^2+\alpha\lVert\mathbf{Z}_{j}\rVert_2^2\right)}{\sum_{i=1}^{N}\exp\left(-\alpha\lVert\mathbf{Z}_{j}-\textbf{E}_{i}\rVert_2^2+\alpha\lVert\mathbf{Z}_{j}\rVert_2^2\right)}\right)\right]\\
    &= -\mathbb{E}_{\mathcal{D}}\left[\log\left(\frac{\exp\left(-\alpha\lVert\mathbf{Z}_{j}-\textbf{E}_{i_j}\rVert_2^2\right)\exp\left(\alpha\lVert\mathbf{Z}_{j}\rVert_2^2\right)}{\sum_{i=1}^{N}\exp\left(-\alpha\lVert\mathbf{Z}_{j}-\textbf{E}_{i}\rVert_2^2\right)\exp\left(\alpha\lVert\mathbf{Z}_{j}\rVert_2^2\right)}\right)\right]\\
    &= -\mathbb{E}_{\mathcal{D}}\left[\log\left(\frac{\exp\left(-\alpha\lVert\mathbf{Z}_{j}-\textbf{E}_{i_j}\rVert_2^2\right)}{\sum_{i=1}^{N}\exp\left(-\alpha\lVert\mathbf{Z}_{j}-\textbf{E}_{i}\rVert_2^2\right)}\right)\right]\\
    &= -\mathbb{E}_{\mathcal{D}}\left[\log\left(\frac{f(\mathbf{Z}_{j}, \textbf{E}_{i_j})}{\sum_{i=1}^{N}f(\mathbf{Z}_{j}, \textbf{E}_{i})}\right)\right]
\end{align*}
\end{proof}
\begin{claim}
\label{claim:3}
With our proposed form of $f(\mathbf{Z}_{j}, \textbf{E}_{i_j})$, minimizing the Tune-InCE loss results in minimizing the distance between the separation embedding $\mathbf{Z}_{j}$ and the corresponding speaker embedding $\textbf{E}_{i_j}$ meanwhile maximizing the distance between other speaker embeddings $\{\textbf{E}_{i} |\forall i\neq i_j\}$.
\end{claim}
\begin{proof}
By substituting Eq. (\ref{eq:1.2}) into Eq. (\ref{eq:1.1}), we have
\begin{align*}
\mathcal{L}_{\tiny\text{InCE}} &= -
    \mathbb{E}_{\mathcal{D}}\left[\log\left(\frac{\exp\left(-\alpha\lVert\mathbf{Z}_{j}- \textbf{E}_{i_j}\rVert_{2}^{2}\right)}{\sum_{i=1}^{N}\exp\left(-\alpha\lVert\mathbf{Z}_{j}- \textbf{E}_{i}\rVert_{2}^{2}\right)}\right)\right]\\
    &= 
    \mathbb{E}_{\mathcal{D}}\left[\alpha\lVert\mathbf{Z}_{j}- \textbf{E}_{i_j}\rVert_{2}^{2}\right]\\
    &\,\,\,\,+\mathbb{E}_{\mathcal{D}}\left[\log\left(\sum_{i=1}^{N}\exp\left(-\alpha\lVert\mathbf{Z}_{j}- \textbf{E}_{i}\rVert_{2}^{2}\right)\right)\right],
\end{align*}
which consists of two terms: (1) the first term is a scaled Euclidean distance with a scalar $\alpha > 0$; (2) the second term is a logarithmic sum of exponentials.
\par
The first term can be used to minimize the Euclidean distance between any speaker embedding and its corresponding separation embeddings. On the contrary, considering the second term, since there is a negative sign on the Euclidean distance, we can see that it is responsible for pulling the separation embedding away from all other speaker embeddings.
\end{proof}

Besides, we can also relate our Tune-InCE loss with some existing works.

\begin{claim}
\label{claim:4}
Tune-InCE loss can be seen as a rescaled $l$-2 normalization of InfoNCE loss proposed in \cite{nce18}.
\end{claim}
\begin{proof}
Our Tune-InCE loss uses a different form of $f(\mathbf{Z}_{j}, \textbf{E}_{i_j})$ than the InfoNCE loss. The latter has
$$f_\text{InfoNCE}(\mathbf{Z}_{j}, \textbf{E}_{i_j})=\exp\left(\mathbf{Z}_{j}^\top \textbf{E}_{i_j}\right),$$
Now, by expanding our proposed $f(\mathbf{Z}_{j}, \textbf{E}_{i_j})$, we get
\begin{align*}
f(\mathbf{Z}_{j}, \textbf{E}_{i_j})&=\exp\left(-\alpha\lVert\mathbf{Z}_{j}- \textbf{E}_{i_j}\rVert_{2}^{2}\right)\\
&=\exp\left(-\alpha(\lVert\mathbf{Z}_{j}\rVert_{2}^{2}+\lVert\textbf{E}_{i_j}\rVert_{2}^{2})+2\alpha\mathbf{Z}_{j}^\top \textbf{E}_{i_j}\right)\\
&=\frac{\exp\left(2\alpha\mathbf{Z}_{j}^\top \textbf{E}_{i_j}\right)}{\exp\left(\alpha\lVert\mathbf{Z}_{j}\rVert_{2}^{2}+\alpha\lVert\textbf{E}_{i_j}\rVert_{2}^{2}\right)}\\
&=\frac{f_\text{InfoNCE}(\mathbf{Z}_{j}, \textbf{E}_{i_j})^{2\alpha}}{\exp\left(\alpha\lVert\mathbf{Z}_{j}\rVert_{2}^{2}+\alpha\lVert\textbf{E}_{i_j}\rVert_{2}^{2}\right)}
\end{align*}
\end{proof}

\subsection{Datasets and Model Setup}
\subsubsection{Datasets}
We used a benchmark 8kHz dataset WSJ0-2mix \cite{hershey2016deep} for comparison with state-of-the-art source separation systems. It consists of $30$ hours of training set comprised of $20000$ utterances from $101$ speakers, $10$ hours of validation set consisting of $5000$ utterances from the same $101$ speakers, and $5$ hours of test data comprised of $3000$ utterances from $18$ speakers unseen in training. We also used a large-scale publicly available 16kHz benchmark dataset Librispeech \cite{panayotov2015librispeech}. Note that our separation task on Librispeech was much more challenging than that on the recently proposed LibriMix dataset \cite{cos2020librimix} because of much fewer utterances per speaker reducing from $55$ to $5$ for training, but our setting was regarded more realistic as it was often hard to collect numerous utterances from the users in real-world applications. 

\subsubsection{Model Setup}
The encoder and decoder structure, as well as the model's hyper-parameter settings, were directly inherited from DPRNN's setup \cite{luo2019dual} for comparison purposes. Note that no model hyper-parameter has been fine-tuned towards our proposed structure, otherwise more improvement could be reasonably expected for ours. We used $6$ consecutive GALR blocks, where $B=4$, $B_1=2$, $B_2=2$ blocks are used for the MSM generic space, the speaker-knowledge space, and the speech-stimuli space, respectively. 

In Table \ref{tab:winlen}, we compared the performance by GALR to that by DPRNN, each with various window length (i.e., number of time-domain samples in each window frame) settings. It is notable that, as the window length gets smaller, the sequence of the window frames that GALR operates on would also become longer and thus increase the computational costs proportionally. In the paper, for WSJ0-2mix, we tried all window length configurations reported in DPRNN and found that the best configuration is to use a 4-sample (5ms) window length for encoding. Correspondingly, we set other correlated hyper-parameters in the Tune-In model as $D=128, K=256, Q=8$. It achieved SOTA results with a remarkable reduction of model size relative to the best configuration (2-sample window) of DPRNN \cite{luo2019dual}. In particular, the best Tune-In model entails $11.5\%$ smaller model size, $62.9\%$ less run-time memory, and $66.4\%$ fewer computational operations. For WSJ0-music, the experiment is conducted for an apple-to-apple comparison, where both DPRNN and GALR are trained on our implementation with equal settings, where a window length of 16 samples and $D=128, K=64, Q=32$ are used for efficient training and evaluation. For Librispeech, due to its different sampling rate, we used a 8-sample (5ms) window length and set $D=128, K=128, Q=16$.

For the hyper-parameters related to the Tune-In losses, we empirically set $\gamma=3$, $\lambda=10$, $\varepsilon=0.05$. In each training epoch, mixture signals lasting $4$s were generated online by masking each clean utterance in the training set with a different random utterance from the same training set at random starting positions, and SIR was randomly selected from a uniform distribution of 0 to 5dB. For testing, mixture signals were pre-mixed with SIR ranging from 0 to 5dB using samples in the test set.

\subsubsection{Training Details}
All the models were trained on NVIDIA Tesla V100 GPU devices using PyTorch. For the separation task, we referred to the training protocol in \cite{luo2019dual}. We use an Adam optimizer \cite{kingma2014adam} with an initial learning rate of $1e^{-3}$ and a weight decaying rate of $1e^{-6}$. The training was considered converged when no lower validation loss can be observed in 10 consecutive epochs. A gradient clipping method was used to ensure the maximum l2-norm of each gradient is less than 5. All models were assessed in terms of SI-SNRi and SDRi \cite{le2019sdr}.

\subsection{Ablation Study on Summarizers Along Inter- and Intra-segment}
Why should self-attention and RNN be considered better summarizers along the inter- and intra-segment dimension, respectively?
Table \ref{tab:summ} gives the dissected results of SI-SNRi on WSJ0-2mix when permuting Bi-LSTM and self-attention model as summarizers along inter- and intra-segment. For efficiency concern, we used a window length of $16$ samples for all systems implemented in this ablation study. The best result is in the lower-left corner and corresponding to the performance of GALR, which is also reported in Table \ref{tab:winlen}. We have also examined replacing Bi-LSTM with CNN as a summarizer within the frame: the best performing Temporal Convolutional Network (TCN) \cite{lea2017temporal} had $18\%$ lower (worse) SI-SNRi, but achieved a speed-up by $26\%$.
\begin{table}[h!]
\centering
\caption{SI-SNRi results on WSJ0-2mix when permuting Bi-LSTM and self-attention model in local and global modeling.}
\label{tab:summ}
\begin{tabular}{c|cc}
\specialrule{.13em}{0em}{0em} 
\makecell{\bf Model} & \makecell{Local RNN} & \makecell{Local Self-Attn} \\
\hline
Global RNN & 15.9 & 12.3  \\
Global Self-Attn & \textbf{17.0} & 14.6 \\
\specialrule{.13em}{0em}{0em} 
\end{tabular}
\end{table}

\subsection{Effectiveness of Dual-Attention}
To study the effectiveness of the proposed Dual-Attention mechanism, we compared it with the recent feature-wise linear modulation (FiLM) \cite{Ethan2017film} method, the latter of which is a solid method that has been successfully applied in both visual reasoning \cite{Ethan2017film} and speech processing domain \cite{zeg2020wavesplit}.
\begin{figure}[bh!]
    \centering
    \begin{tikzpicture}
    \begin{axis}[
        xlabel={Epoch},
        ylabel={valid SI-SNR},
        ymin=0, ymax=15.5,
        legend pos=south east,
        ymajorgrids=true,
        xmajorgrids=true,
        grid style=dashed,
        legend cell align={left},
    ]
    
    \addplot[
        color=red,
        ]
        coordinates {
            (1,0.02)
            (2,5.68)
            (3,9.21)
            (4,10.37)
            (5,10.88)
            (6,11.11)
            (7,11.36)
            (8,11.50)
            (9,11.71)
            (10,11.80)
            (11,11.79)
            (12,11.81)
            (13,12.02)
            (14,12.00)
            (15,12.07)
            (16,12.19)
            (17,12.21)
            (18,12.24)
            (19,12.36)
            (20,12.28)
            (21,12.39)
            (22,12.32)
            (23,12.40)
            (24,12.37)
            (25,12.47)
            (26,12.53)
            (27,12.51)
            (28,12.52)
            (29,12.58)
            (30,12.61)
            (31,12.62)
            (32,12.52)
            (33,12.55)
            (34,12.74)
            (35,12.69)
            (36,12.67)
            (37,12.68)
            (38,12.71)
            (39,12.78)
            (40,12.77)
            (41,12.79)
            (42,12.83)
            (43,12.76)
            (44,12.80)
            (45,12.86)
            (46,12.84)
            (47,12.86)
            (48,12.92)
            (49,12.84)
            (50,12.93)
            (51,12.96)
            (52,12.91)
            (53,12.95)
            (54,12.97)
            (55,13.01)
            (56,13.05)
            (57,13.00)
            (58,13.00)
            (59,13.01)
            (60,13.10)
            (61,13.05)
            (62,13.10)
            (63,13.13)
            (64,13.15)
            (65,13.09)
            (66,13.15)
            (67,13.16)
            (68,13.16)
            (69,13.19)
            (70,13.10)
            (71,13.24)
            (72,13.15)
            (73,13.06)
            (74,13.18)
            (75,13.16)
            (76,13.21)
            (77,13.20)
            (78,13.20)
            (79,13.24)
            (80,13.21)
        };
    \addlegendentry{FiLM}
    \node[label={-90:{13.24}},color=red,circle,fill,inner sep=1.3pt] at (axis cs:71,13.24) {};
    \addplot[
        color=purple,
        ]
        coordinates {
            (1,7.78)
            (2,8.58)
            (3,9.24)
            (4,8.22)
            (5,9.61)
            (6,9.80)
            (7,9.24)
            (8,9.77)
            (9,9.63)
            (10,10.11)
            (11,9.93)
            (12,10.14)
            (13,10.10)
            (14,9.92)
            (15,10.10)
            (16,9.93)
            (17,10.02)
            (18,10.58)
            (19,10.27)
            (20,10.39)
            (21,10.60)
            (22,10.29)
            (23,10.51)
            (24,10.22)
            (25,10.42)
            (26,10.75)
            (27,10.76)
            (28,10.51)
            (29,10.79)
            (30,10.54)
            (31,10.59)
            (32,10.64)
            (33,10.70)
            (34,10.74)
            (35,11.03)
            (36,10.83)
            (37,10.69)
            (38,10.87)
            (39,10.75)
            (40,10.54)
            (41,10.78)
            (42,10.77)
            (43,10.90)
            (44,10.76)
            (45,10.71)
        };
    \addlegendentry{GA with FiLM}
    \node[label={-90:{11.03}},color=purple,circle,fill,inner sep=1.3pt] at (axis cs:35,11.03) {};
    \addplot[
        color=blue,
        ]
        coordinates {
            (1.00,0.76)
            (2.00,8.36)
            (3.00,11.03)
            (4.00,11.88)
            (5.00,12.29)
            (6.00,12.47)
            (7.00,12.61)
            (8.00,12.84)
            (9.00,12.50)
            (10.00,12.80)
            (11.00,12.89)
            (12.00,12.92)
            (13.00,13.01)
            (14.00,12.95)
            (15.00,13.12)
            (16.00,13.16)
            (17.00,13.14)
            (18.00,13.24)
            (19.00,13.26)
            (20.00,13.07)
            (21.00,13.15)
            (22.00,13.26)
            (23.00,13.20)
            (24.00,13.09)
            (25.00,13.43)
            (26.00,13.27)
            (27.00,13.21)
            (28.00,13.39)
            (29.00,13.41)
            (30.00,13.32)
            (31.00,13.30)
            (32.00,13.42)
            (33.00,13.41)
            (34.00,13.06)
            (35.00,13.52)
            (36.00,13.34)
            (37.00,13.52)
            (38.00,13.45)
            (39.00,13.42)
            (40.00,13.44)
            (41.00,13.46)
            (42.00,13.46)
            (43.00,13.42)
            (44.00,13.50)
            (45.00,13.52)
            (46.00,13.36)
            (47.00,13.49)
            (48.00,13.60)
            (49.00,13.58)
            (50.00,13.60)
            (51.00,13.70)
            (52.00,13.62)
            (53.00,13.53)
            (54.00,13.67)
            (55.00,13.55)
            (56.00,13.67)
            (57.00,13.77)
            (58.00,13.66)
            (59.00,13.58)
            (60.00,13.67)
            (61.00,13.67)
            (62.00,13.65)
            (63.00,13.59)
            (64.00,13.61)
            (65.00,13.72)
            (66.00,13.53)
            (67.00,13.58)
            (68.00,13.87)
            (69.00,13.69)
            (70.00,13.75)
            (71.00,13.86)
            (72.00,13.80)
            (73.00,13.83)
            (74.00,13.73)
            (75.00,13.82)
            (76.00,13.66)
            (77.00,13.74)
        };
    \addlegendentry{DualAttn}
    \node[label={90:{13.87}},color=blue,circle,fill,inner sep=1.3pt] at (axis cs:68,13.87) {};
    \end{axis}
    \end{tikzpicture}
    \caption{
    Validation SI-SNR comparison among 1. FiLM, 2. GA with FiLM, i.e., FiLM inside globally attentive layer, and 3. dual attention (DualAttn). The converged best valid SI-SNR scores are marked as dots for each method.}
\label{fig:valid}
\end{figure}
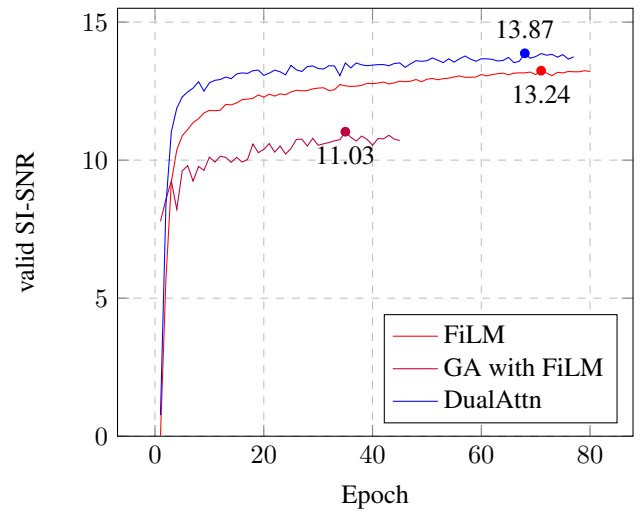

The Dual-Attention mechanism took the form of the following:
\begin{align}
\mathbf{A}^{(b)}_j&=\text{LayerNorm}(r(\mathbf{Z}_j)\odot\mathbf{G}^{(b)}+ h(\mathbf{Z}_j)) \label{eq:adualattn}\\
    \mathbf{a}_{\textbf{dual},j}&=\text{softmax}\left(\text{Query}(\mathbf{G}^{(b)})^\top \text{Key}(\mathbf{A}^{(b)}_j)\right)\\
    \hat{\mathbf{X}}^{(b+1)}_j&=\sum_S \mathbf{a}_{\textbf{dual},j}\cdot \text{Value}(\mathbf{A}^{(b)}_j),
\end{align}
where $\mathbf{A}^{(b)}_j$ was modified from the original FiLM formula with $r(\cdot)$ and $h(\cdot)$ being two learnable linear mappings.

For comparison, the FiLM mechanism applying to the information flow among the GALR cells took the form of the following:
\begin{align}
    \mathbf{a}_{\textbf{dual},j}&=\text{softmax}\left(\text{Query}(\mathbf{G}^{(b)})^\top \text{Key}(\mathbf{G}^{(b)})\right)\\
    \bar{\mathbf{X}}^{(b+1)}&=\sum_S \mathbf{a}_{\textbf{dual},j}\cdot \text{Value}(\mathbf{G}^{(b)}),\\
    \hat{\mathbf{X}}^{(b+1)}_j&=\text{PReLU}(r(\mathbf{Z}_j)\odot\bar{\mathbf{X}}^{(b+1)}+ h(\mathbf{Z}_j)),
\end{align}

For substantial comparison, we also investigated an alternative way of applying FiLM to modulate information in the speech stimuli space, where FiLM took effect inside the globally attentive layer (namely ``GA with FiLM") as follows.
\begin{align}
\mathbf{F}^{(b)}_j&=\text{PReLU}(r(\mathbf{Z}_j)\odot\mathbf{G}^{(b)}+ h(\mathbf{Z}_j))\label{eq:maskedF}\\
    \mathbf{a}_{\textbf{dual},j}&=\text{softmax}\left(\text{Query}(\mathbf{G}^{(b)})^\top \text{Key}(\mathbf{F}^{(b)}_j)\right)\label{eq:filmsoftmax}\\
    \hat{\mathbf{X}}^{(b+1)}_j&=\sum_S \mathbf{a}_{\textbf{dual},j}\cdot \text{Value}(\mathbf{F}^{(b)}_j),\label{eq:filmout}
\end{align}

As shown in Fig. \ref{fig:valid}, the proposed Dual-Attention method achieved substantial improvements in terms of sustainable training stability as well as a higher converged validation SI-SNR score. In comparison, despite a worse SI-SNR score, FiLM is still a solid method when applying to the intermediate representations in between each GALR cell. However, when applying FiLM inside the GALR cell (GA with FiLM), the training stability and performance notably dropped.

Consequently, we inspected the underlying key factor that leads to the above performance gap. Since $\mathbf{Z}_j$, the steering vector from the speaker-knowledge space, was computed from the cross attention w.r.t. to the $j$-th speaker, the $r(\cdot)$ and $h(\cdot)$ together modulated $\mathbf{G}^{(b)}$ to condition on the $j$-th speaker in the speech-stimuli space. In this case, we conceive that the ReLU activation function proposed in the original FiLM is not suitable here, since the masked $\mathbf{F}^{(b)}_j$ in Eq. \ref{eq:maskedF} could induce high sparsity in the softmax matrix $\mathbf{a}_{\textbf{dual},j}$ in Eq. \ref{eq:filmsoftmax}. Followed from a sparse $\mathbf{a}_{\textbf{dual},j}$, the thresholding effect inside an attention mechanism would limit the learning capacity of the globally attentive layer and is prone to over-fitting. 

As shown in Fig. \ref{fig:DualAttn}, an effective method should have obtained relatively higher softmax scores on current time points than other time points since the two attention targets ($\mathbf{F}^{(b)}_j$ v.s. $\mathbf{G}^{(b)}$) were aligned in self-attention. We can observe from the heat map in Fig. \ref{fig:DualAttn} that, in the case of ``GA with FiLM", the softmax matrix appears to be relatively sparse as some values become constants after the element-wise PReLU activation. The diagonal weights are also weak since the attention model was forced to be ``focused" only on limited targets.
In contrast, as seen from the heat map, the proposed Dual-Attention technique produced a more interpretable and sensible emphasis of diagonal attention weights. Moreover, the Dual-Attention itself can be viewed as a non-linear activation w.r.t. $\mathbf{A}^{(b)}_j$ in Eq. \ref{eq:adualattn}, which essentially remedied the usage of ReLU in FiLM \cite{Ethan2017film} yet was more effective when applied to our Tune-In system.
\begin{figure*}[ht]
    \centering 
    \begin{minipage}{0.03\textwidth}
        \large\rotatebox{90}{GA with FiLM}
    \end{minipage}
    \begin{minipage}{0.96\textwidth}
    \begin{subfigure}{.49\textwidth}
      \centering
      \includegraphics[width=1.1\linewidth]{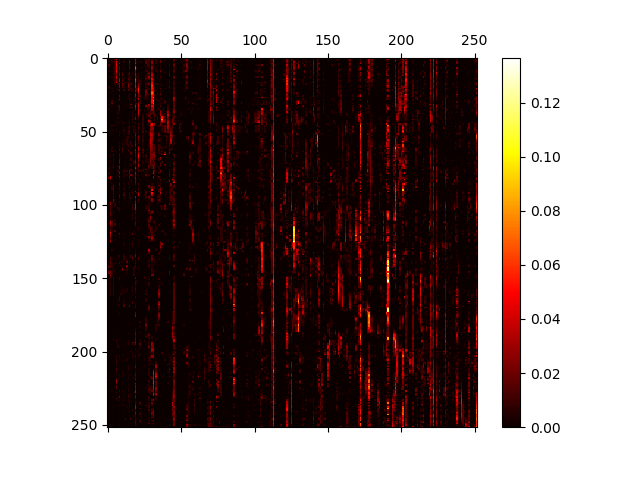}  
    \end{subfigure}
    \begin{subfigure}{.49\textwidth}
      \centering
      \includegraphics[width=1.1\linewidth]{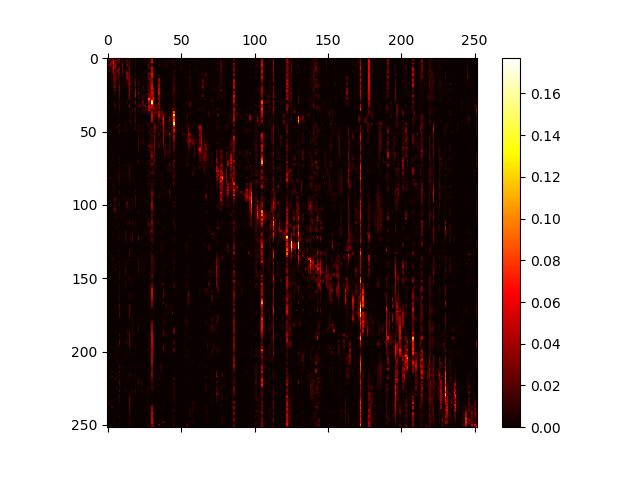}  
    \end{subfigure}
    \end{minipage}
    \begin{minipage}{0.03\textwidth}
        \large\rotatebox{90}{DualAttn}
    \end{minipage}
    \begin{minipage}{0.96\textwidth}
    \begin{subfigure}{.49\textwidth}
      \centering
      \includegraphics[width=1.1\linewidth]{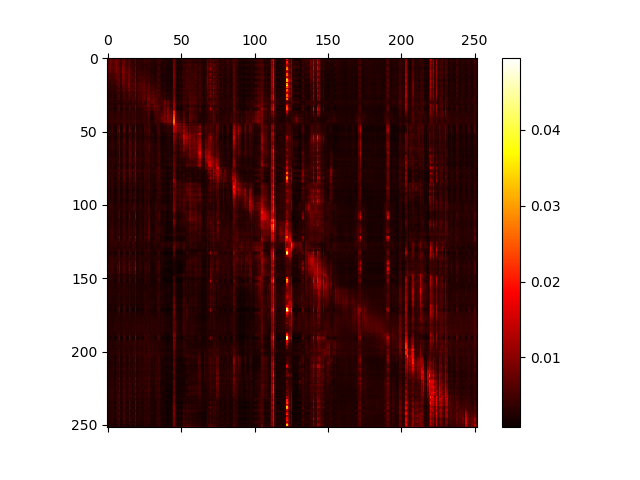}  
      \caption{The first speaker}
    \end{subfigure}
    \begin{subfigure}{.49\textwidth}
      \centering
      \includegraphics[width=1.1\linewidth]{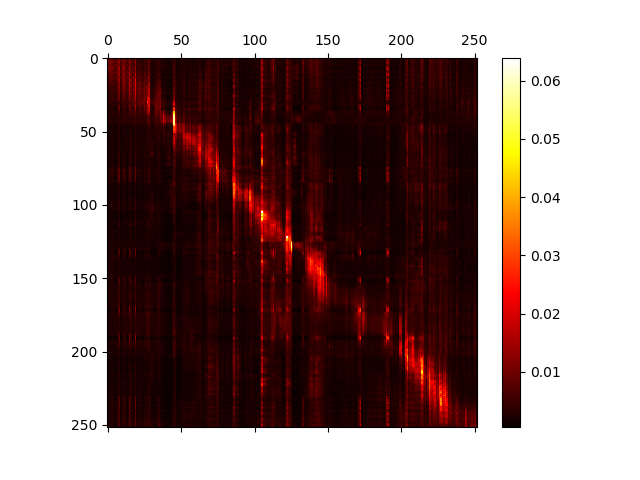}  
      \caption{The second speaker}
    \end{subfigure}
    \end{minipage}
    \caption{Averaged softmax weights of two concurrent speakers in dual attention base on GA with FiLM (1st row) and DualAttn (2nd row). In the case of DualAttn, relatively strong attention weights can be observed at the diagonal.}
    \label{fig:DualAttn}
\end{figure*}

\subsection{Automatically Learnt Stimuli-based Selective Cross Attention}
Auditory selective attention has been widely studied \cite{mesga2012nature, costa2013tunein, obleser2015selective, sullivan2015eeg} in behavioral and cognitive neurosciences. It is about the manner that a human is not able to listen to, or remember two concurrent speech streams, while listeners usually select the attended speech and ignore other sounds in a complex auditory scene. 

Although our proposed Tune-in system takes no regularization regarding the above manner, we observed an interesting phenomenon that a similar selective bottom-up cross attention could be automatically learnt based on the stimuli. We plot the cross-attention curves in Fig. \ref{fig:male_female}-\ref{fig:CrossAttn1} by averaging the softmax over the $S$ embedding length to obtain an $S_j$-dimensional attention vector for each speaker $j$. The curves were generated from an online mode, in which case $S_j$ was equal to $S$. By aligning the $S$ segments (each of size $K$, and in all $S*K*M/4$ samples) along with the time axis, we plotted the $C$ attention curves along with the raw input signal. As shown in Fig. \ref{fig:male_female}-\ref{fig:CrossAttn1}, at any given time segment, it is generally the most salient target in the mixture that triggers the corresponding cross-attention curve. In places where both sources are soft, both cross-attention curves are low to ignore this place as it could be noisy and unreliable. Notably, there is hardly any time point where both cross-attention curves are raised. 

Note that a human can also perform top-down auditory selective attention, such as that based on a task-relevant stimulus. To discriminate from this, here we call our mechanism a stimuli-based bottom-up selective cross attention because the selection is purely based on information from the bottom-up signal, and we leave top-down selective attention for future research work.
\begin{figure*}[ht]
\label{fig:4plots}
      \begin{subfigure}{0.49\textwidth}
      \centering
      \includegraphics[width=\linewidth]{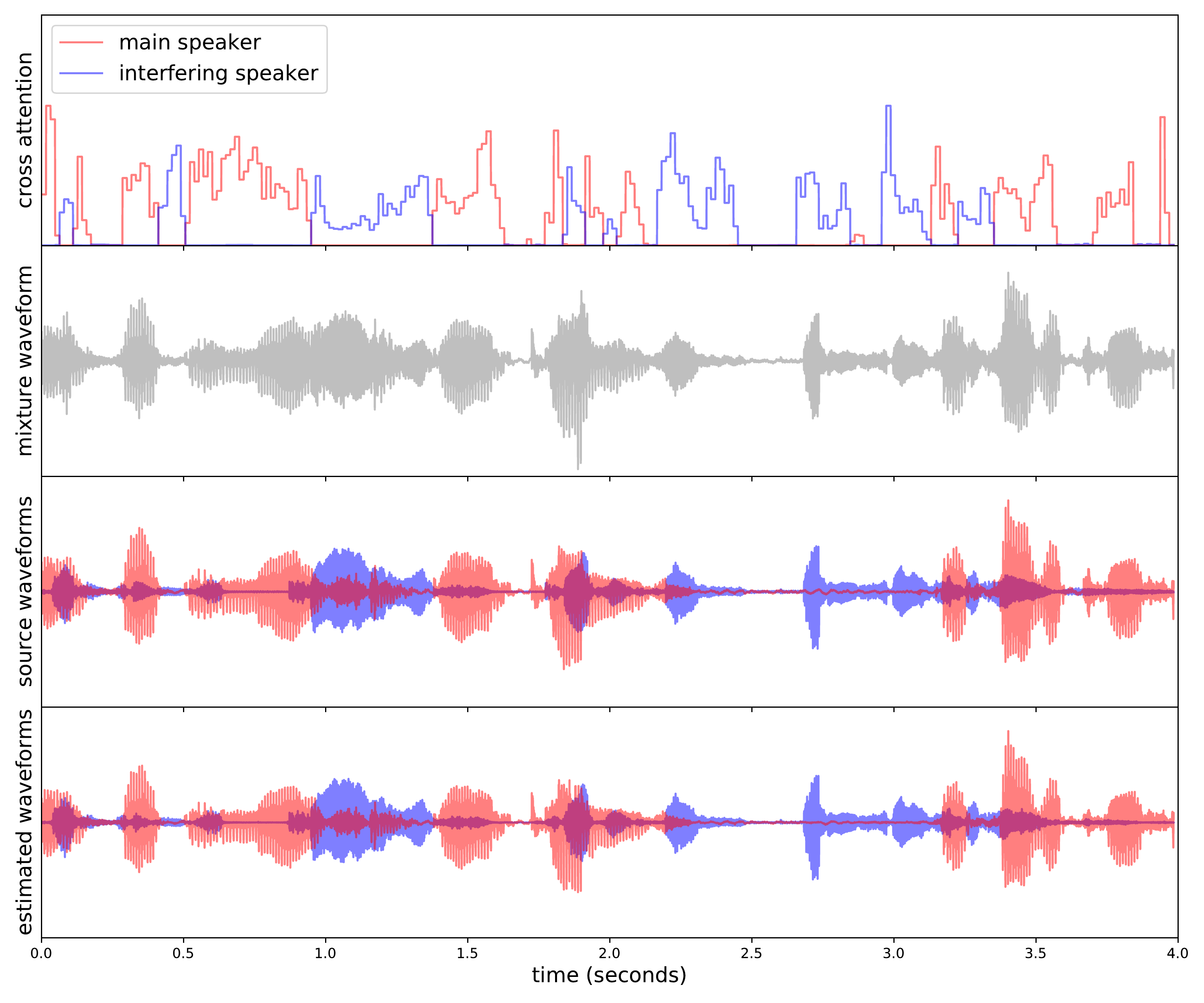}         
      \caption{A male speaker interfered by a female}
      \label{fig:male_female}
      \end{subfigure}
    \begin{subfigure}{0.49\textwidth}
      \centering
       \includegraphics[width=\linewidth]{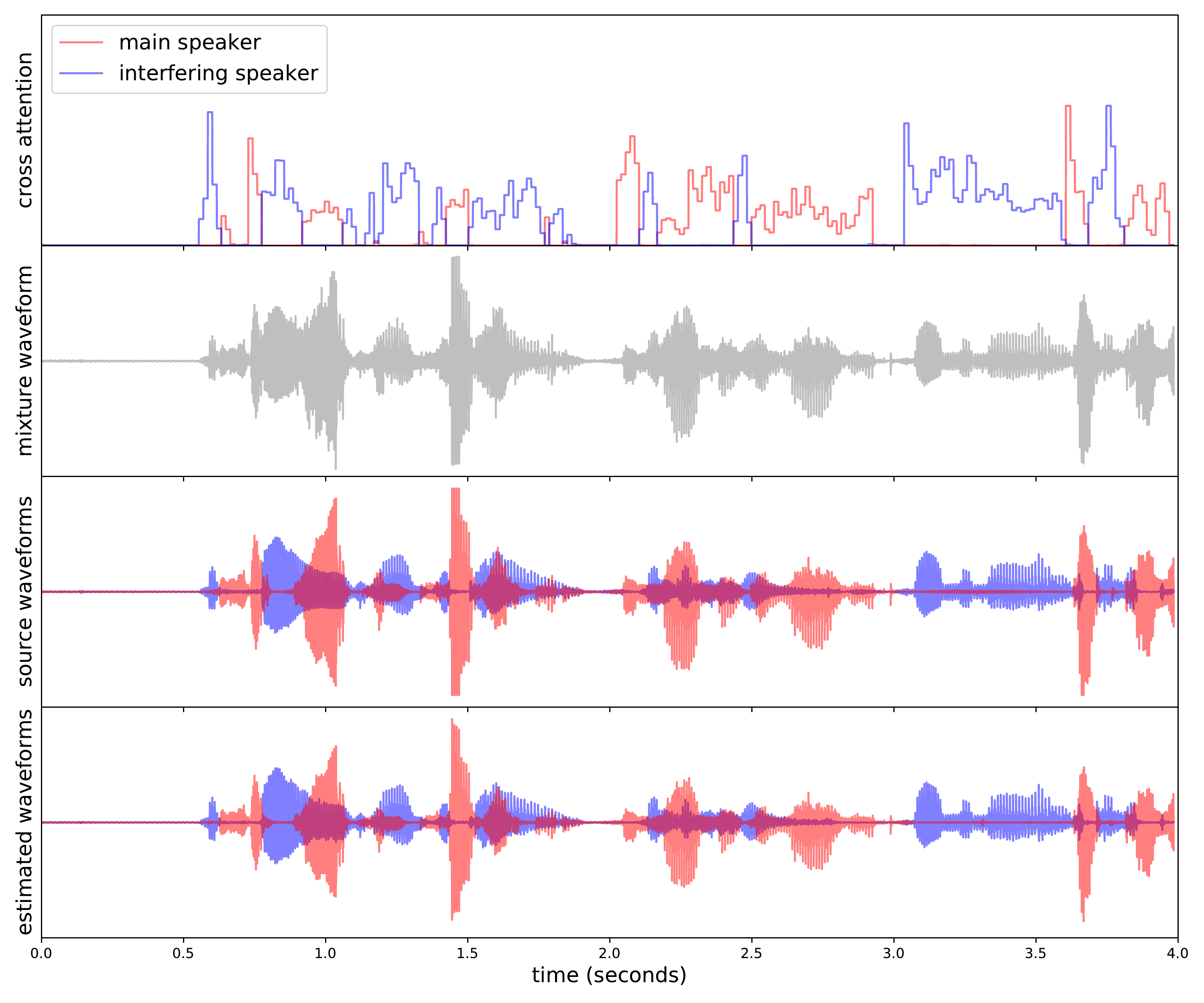} 
       \caption{A male speaker interfered by another male}  
    \end{subfigure}
      \begin{subfigure}{0.49\textwidth}
      \centering
      \includegraphics[width=\linewidth]{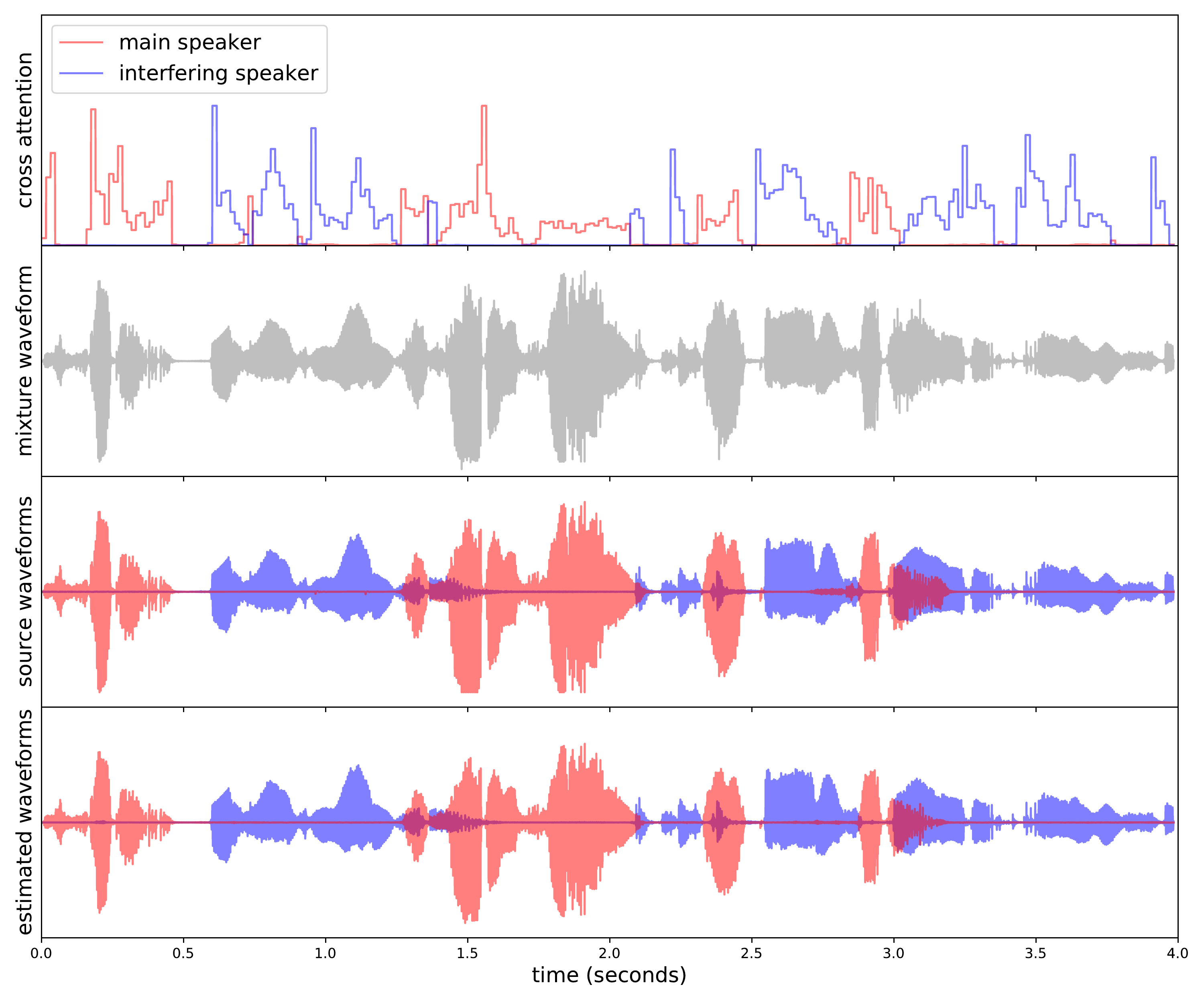} 
      \caption{A female speaker interfered by another female}
      \end{subfigure}
      \begin{subfigure}{0.49\textwidth}
      \centering
      \includegraphics[width=\linewidth]{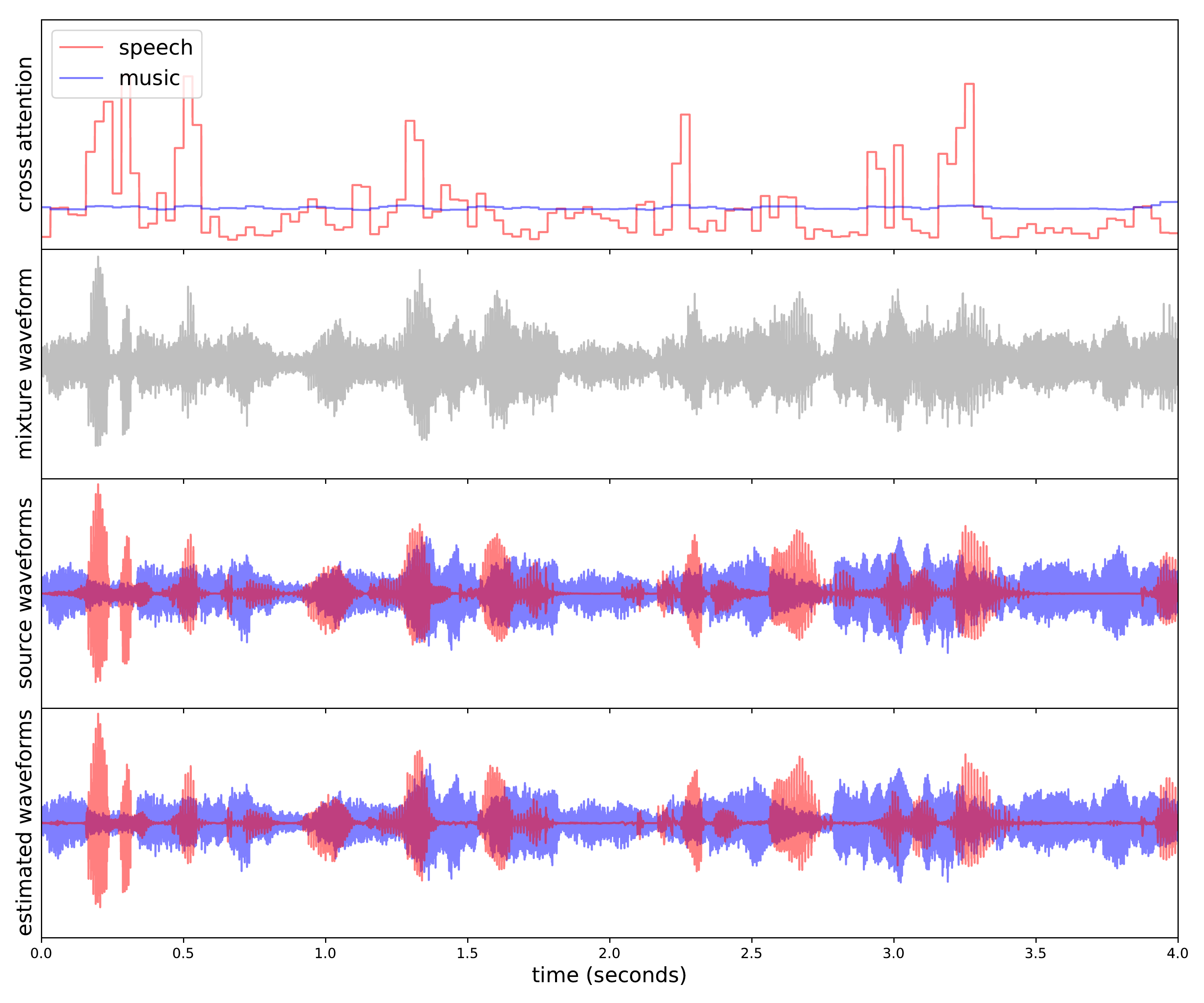} 
      \caption{A male speaker interfered by music}
      \label{fig:CrossAttn1}
      \end{subfigure}
\caption{Demonstrations of selective bottom-up cross attentions automatically learnt based on the speech stimuli.}
\end{figure*}

\subsection{Generalizability in Comparison to Supervised Learning}
\label{sec:study}

\begin{figure*}[t!]
 \centering
    \includegraphics[width=0.32\linewidth]{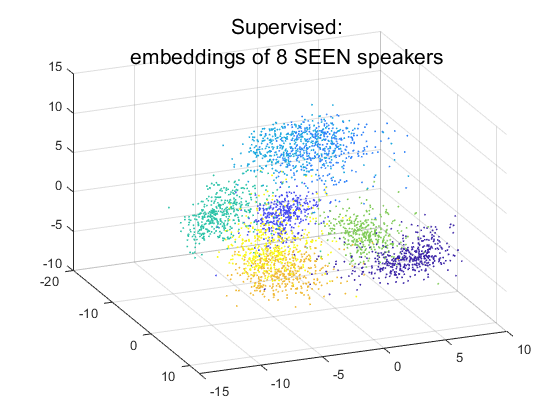}
    \includegraphics[width=0.32\linewidth]{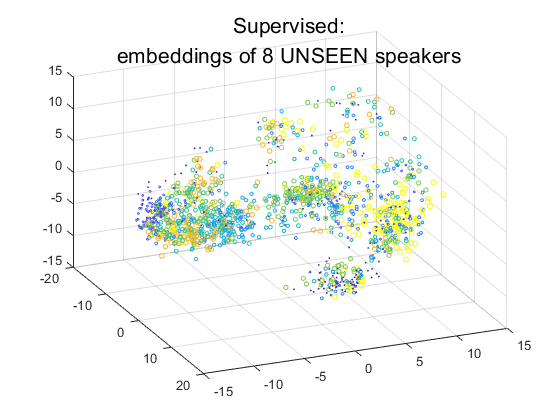}
    \includegraphics[width=0.32\linewidth]{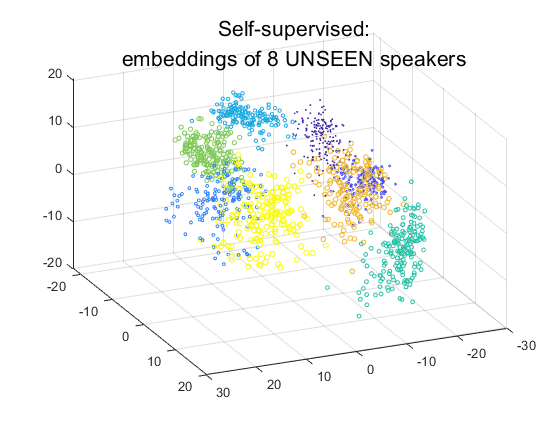}
\caption{3-D PCA of the speaker embeddings, left: from 8 random seen speakers, by supervised learning, middle: from 8 random unseen speakers, by supervised learning, and right: from the same 8 unseen speakers, by self-supervised learning.}
\label{fig:pca}
\end{figure*}

We analyze the generalizability of the deep representations learnt by the proposed self-supervised loss ``\text{Tune-InCE}" and compare it to that by the supervised approach ``\text{Tune-ID}". Fig. \ref{fig:pca} shows the projection of the speaker embeddings to a 3-D PCA space, where the same color indicates the same speaker. By supervised learning ``\text{Tune-ID}", despite the well discriminative embeddings learnt for the training ``seen" speakers (in the right in Fig. \ref{fig:pca}), the discriminative power dropped drastically for ``unknown" speakers, as shown in the middle of Fig. \ref{fig:pca}. On the contrary, our proposed approach ``\text{Tune-InCE}" can extract substantially discriminative embeddings for both seen and unknown speakers, as shown in the right of Fig. \ref{fig:pca}. It turns out the self-supervised learning purges our model from learning a trivial task of speaker identity prediction, but instead enforces it to learn deep representations with essential discriminative power and generalization capability.

\subsection{Speaker-Embedding-Based Permutation Computation for Training Speedup}
\label{sec:pit}
Noted that for computing the speech loss $\mathcal{L}_\text{SI-SNR}$, we need to assign the correct reference (or target) source signals; likewise, for computing the speaker loss in Eq. (\ref{eq:1.1}), we need to assign the correct corresponding speaker vectors. However, reference assigning has ambiguity since the model gives multiple outputs, one for each source, and they depend on the same input mixture. This problem is referred to as the permutation problem and has been properly solved via the utterance-level permutation invariant training (u-PIT) method \cite{yu2017permutation}.

Due to our proposed dual-attention mechanism, once the output permutation in one space is resolved, the permutation in the other space can be determined accordingly. During training, we start with using u-PIT to calculate the speech reconstruction loss $\mathcal{L}_\text{SI-SNR}$ in the speech-stimuli space. After an empirical number of epochs when the speaker vectors $\textbf{E}_i$ have reached to a relatively stable state, we switch to using u-PIT to calculate the speaker loss instead, and the determined assignment is used to ``steer" the signal reconstruction in the speech-stimuli space. In the ``steered" phase, the speech reconstruction loss becomes PIT-free and thus relieved from the relatively heavier computation burden.

\end{document}